\PassOptionsToPackage{table}{xcolor}
\documentclass{article}
\usepackage{fontawesome,hhline,multirow}
\usepackage{amsmath,amssymb,amsthm}
\usepackage[hidelinks]{hyperref}
\hypersetup{breaklinks=true}
\usepackage{xurl}

\usepackage{libertine}
\usepackage[varqu]{zi4}
\usepackage[libertine]{newtxmath}
\usepackage[T1]{fontenc}
\usepackage[mathscr]{euscript}
\usepackage[utf8]{inputenc}

\usepackage{microtype}

\theoremstyle{plain}
\newtheorem{theorem}{Theorem}[section]
\newtheorem{proposition}[theorem]{Proposition}
\newtheorem{corollary}[theorem]{Corollary}
\newtheorem{lemma}[theorem]{Lemma}
\newtheorem{problem}[theorem]{Problem}
\theoremstyle{definition}
\newtheorem{definition}[theorem]{Definition}
\theoremstyle{remark}
\newtheorem{example}[theorem]{Example}
\newtheorem{remark}[theorem]{Remark}

\usepackage[margin=1.25in]{geometry}

\newcommand{\Cluster}{\mathcal{C}}
\newcommand{\NonFaulty}[1]{\mathop{\textnormal{\textsf{nf}}}(#1)}
\newcommand{\Faulty}[1]{\mathop{\textnormal{\textsf{f}}}(#1)}
\newcommand{\n}[1]{\mathbf{n}_{#1}}
\newcommand{\f}[1]{\mathbf{f}_{#1}}
\newcommand{\nf}[1]{\mathbf{nf}_{#1}}
\newcommand{\Replica}[1][r]{\textnormal{\textsc{#1}}}

\newcommand{\Decide}[1]{\textnormal{\textsc{#1}}}
\newcommand{\Sign}[2]{\langle #1 \rangle_{#2}}
\newcommand{\SignMessage}[3]{\Sign{\textnormal{\texttt{#1}} : #2}{#3}}

\newcommand{\SF}{\Phi}
\newcommand{\SFi}{\SF_{\min}}
\newcommand{\SFa}{\SF_{\max}}
\newcommand{\FaultyPos}[2]{\lVert #1; #2 \rVert_{\mathbf{f}}}
\newcommand{\List}[1]{\mathop{\mathsf{list}}(#1)}
\newcommand{\Permute}[1]{\mathop{\mathsf{perms}}(#1)}
\newcommand{\Repeat}[2]{#2^{:#1}}
\newcommand{\FC}[4]{\mathbb{F}(#1, #2, #3, #4)}
\newcommand{\PT}[3]{\mathbb{E}(#1, #2, #3)}
\newcommand{\Limit}[2]{#1|_{#2}}
\newcommand{\LM}[2]{\mathbb{M}(#1, #2)}

\newcommand{\Name}[1]{\textnormal{\textsc{#1}}}

\newcommand{\abs}[1]{\lvert #1 \rvert}
\newcommand{\union}{\cup}

\newcommand{\difference}{\setminus}

\newcommand{\BigOm}[1]{\Omega(#1)}
\newcommand{\dsfrac}[2]{#1/#2}
\newcommand{\subref}[2]{\ref{#1}(\ref{#1:#2})}
\renewcommand{\div}{\operatorname{div}}
\newcommand{\lfref}[2]{Line~\ref{#1:#2} of Figure~\ref{#1}}
\newcommand{\lsfref}[3]{Lines~\ref{#1:#2}--\ref{#1:#3} of Figure~\ref{#1}}

\newcommand{\Good}[1]{\cellcolor{green!20}#1}
\newcommand{\Mid}[1]{\cellcolor{yellow!20}#1}
\newcommand{\Bad}[1]{\cellcolor{red!20}#1}
\newcommand{\BadT}[1]{\color{red!75}\underline{#1}}
\newcommand{\CheckMark}{\Good{\faCheck}}
\newcommand{\FailMark}{\Bad{\faClose}}

\usepackage[noend]{algorithmic}
\newcommand{\GETS}{:=}
\newcommand{\PROTOCALL}[2]{\textnormal{\Name{#1}(#2)}}
\newenvironment{myprotocol}[1]{
    \noindent\textbf{Protocol} {#1}\textbf{:}
    \smallskip
    \hrule
    \smallskip
    \begin{algorithmic}[1]
        \newcommand{\SPACE}{\item[]}
        \newcommand{\TITLE}[1]{\item[] \textbf{\underline{##1}:}\\[2pt]}

        \makeatletter
            \newcommand{\EVENT}[1]{\STATE \textbf{event} ##1 \textbf{do}\begin{ALC@g}}
            \newcommand{\ENDEVENT}{\end{ALC@g}}
        \makeatother
}{
    \end{algorithmic}
    \smallskip
    \hrule
    \smallskip
}

\usepackage{tikz,pgfplots,pgfplotstable}
\usetikzlibrary{arrows.meta,decorations.pathreplacing}
\tikzset{
    plot/.append style={baseline,scale=0.6},
    dot/.append style={circle,scale=0.35,draw=black,fill=black},
    label/.append style={align=center,font=\strut\footnotesize},
    >=Stealth
}
\definecolor{colA}{RGB}{230,159,0}
\definecolor{colB}{RGB}{86,180,233}
\definecolor{colC}{RGB}{0,158,115}
\definecolor{colD}{RGB}{240,228,66}
\definecolor{colE}{RGB}{0,114,178}
\definecolor{colF}{RGB}{213,94,0}
\definecolor{colG}{RGB}{204,121,167}

\pgfplotscreateplotcyclelist{mycyclelist}{
    thick,solid,colA,every mark/.append style={solid},mark=*\\    
    thick,solid,colC,every mark/.append style={solid},mark=*\\    
    thick,solid,colE,every mark/.append style={solid},mark=*\\    
    thick,solid,colG,every mark/.append style={solid},mark=*\\    
    thick,solid,black,every mark/.append style={solid},mark=*\\   
}

\pgfplotsset{
    compat=1.14,
    tick label style={font=\large},
    legend style={font=\Large,cells={anchor=west}},
    title style={font=\Large},
    label style={font=\Large},
    width=270pt,
    height=207pt,
    enlargelimits=0.05,
    every axis/.append style={
        ylabel near ticks,
        mark size=0.75pt,
        cycle list name=mycyclelist,
        font=\Large
    }
}

\pgfplotstableread{
c	f	pcs	plcs
2	0	1.0	1.0
2	1	2.25	2.5
2	2	2.7777777777	3.1666666666
2	3	3.0625	3.45
2	4	3.24	3.5857142857
2	5	3.3611111111	3.6626984126
2	6	3.4489795918	3.7132034632
2	7	3.515625	3.7497086247
2	8	3.5679012345	3.7777000777
2	9	3.61	3.7999794323
2	10	3.6446280991	3.8181764056
2	11	3.6736111111	3.8333319157
2	12	3.698224852	3.8461534763
2	13	3.7193877551	3.8571427609
2	14	3.7377777777	3.8666666417
2	15	3.75390625	3.8749999935
2	16	3.7681660899	3.8823529395
2	17	3.7808641975	3.8888888884
2	18	3.7922437673	3.8947368419
2	19	3.8025	3.8999999999
2	20	3.8117913832	3.9047619047
}\dataCaseTwo

\pgfplotstableread{
c	f	pcs	plcs
3	0	1.0	1.0
3	1	1.7777777777	1.8333333333
3	2	1.96	2.0111111111
3	3	2.0408163265	2.0827380952
3	4	2.086419753	2.1212409812
3	5	2.1157024793	2.1452978238
3	6	2.1360946745	2.1617647641
3	7	2.1511111111	2.1737487075
3	8	2.1626297577	2.1828629545
3	9	2.1717451523	2.1900289855
3	10	2.1791383219	2.195811618
3	11	2.1852551984	2.2005764414
3	12	2.1904	2.2045705861
3	13	2.1947873799	2.2079671286
3	14	2.1985731272	2.2108909086
3	15	2.2018730489	2.2134342444
3	16	2.2047750229	2.2156668792
3	17	2.2073469387	2.2176424859
3	18	2.2096420745	2.2194030448
3	19	2.211702827	2.220981864
3	20	2.2135633551	2.2224057091
}\dataCaseThree

\title{Byzantine Cluster-Sending in Expected Constant Communication}
\author{Jelle Hellings\footnotemark[1]{ }\footnotemark[2] \and Mohammad Sadoghi\footnotemark[2]}

\date{\footnotemark[1]~Department of Computing and Software, McMaster University\\\makebox[0pt]{\footnotemark[2]~~Exploratory Systems Lab, Department of Computer Science, University of California, Davis}}

\begin{document}

\maketitle

\begin{abstract}
Traditional resilient systems operate on fully-replicated fault-tolerant clusters, which limits their scalability and performance. One way to make the step towards resilient high-performance systems that can deal with huge workloads, is by enabling independent fault-tolerant clusters to efficiently communicate and cooperate with each other, as this also enables the usage of high-performance techniques such as sharding and parallel processing. Recently, such inter-cluster communication was formalized as the \emph{Byzantine cluster-sending problem}, and worst-case optimal protocols have been proposed that solve this problem. Unfortunately, these protocols have an all-case \emph{linear complexity} in the size of the clusters involved. 

In this paper, we propose \emph{probabilistic cluster-sending techniques} that can reliably send messages from one Byzantine fault-tolerant cluster to another with only an \emph{expected constant message complexity}, this independent of the size of the clusters involved. Depending on the robustness of the clusters involved, our techniques require only  \emph{two-to-four} message round-trips. Furthermore, our protocols can support worst-case linear communication between clusters, which is optimal, and deal with asynchronous and unreliable communication. As such, our work provides a strong foundation for the further development of resilient high-performance systems.
\end{abstract}

\section{Introduction}
The promises of \emph{resilient data processing}, as provided by private and public blockchains~\cite{book,blockchain_dist,bit_pedigree}, has renewed interest in traditional consensus-based Byzantine fault-tolerant resilient systems~\cite{wild,pbftj,paxossimple}. Unfortunately, blockchains and other consensus-based systems typically rely on fully-replicated designs, which limits their scalability and performance. Consequently, these systems cannot deal with the ever-growing requirements in data processing~\cite{hypereal,idc}.

We believe that \emph{cluster-sending protocols}---which provide reliable communication \emph{between} Byzantine fault-tolerant clusters---have a central role towards bridging \emph{resilient} and \emph{high-performance} data processing. To illustrate this, we refer to the system designs in Figure~\ref{fig:example_intro}. In the traditional design on the \emph{left}, resilience is provided by a fully-replicated Byzantine fault-tolerant cluster, coordinated by some consensus protocol, that holds all data and process all requests. This traditional design has only limited performance, even with the best consensus protocols, and lacks scalability. To improve on the design of traditional systems,  one can employ the \emph{sharded} design on the \emph{right}. In this design, each cluster only holds part of the data. Consequently, each cluster only needs to process requests that affect data they hold. In this way, this sharded design improves performance by enabling \emph{parallel processing} of requests by different clusters, while also improving storage scalability.

\begin{figure}[h!]
\centering
\quad
\begin{tikzpicture}[scale=0.95]
    \filldraw[very thick,fill=yellow!10,draw=yellow!10!black!30] (-0.35, -0.35) rectangle (2.35, 2.35);
    \filldraw[very thick,fill=orange!10,draw=orange!10!black!30] (-0.25, -0.25) rectangle (2.25, 2.25);
    \node (r1) at (0, 2) {$\Replica_1$};
    \node (r2) at (2, 2) {$\Replica_2$};
    \node (r3) at (0, 0) {$\Replica_3$};
    \node (r4) at (2, 0) {$\Replica_4$};
    \path[<->,thin] (r1) edge (r2) edge (r3) edge (r4)
                    (r2) edge (r3) edge (r4)
                    (r3) edge (r4);
    \node[above,align=center] at (1, 2.25) {Cluster\\(All Data)\strut};
    \node[below,align=center] (ec) at (1, -0.65) {Requests\\(All Data)\strut} edge[thick,->] (1,-0.25);
\end{tikzpicture}
\hfill
\begin{tikzpicture}[scale=0.95]
    \filldraw[very thick,fill=yellow!10,draw=yellow!10!black!30] (-0.35, -0.35) rectangle (8.35, 2.35);
    \filldraw[very thick,fill=orange!10,draw=orange!10!black!30] (-0.25, -0.25) rectangle (2.25, 2.25);
    \node (r11) at (0, 2) {$\Replica[e]_1$};
    \node (r12) at (2, 2) {$\Replica[e]_2$};
    \node (r13) at (0, 0) {$\Replica[e]_3$};
    \node (r14) at (2, 0) {$\Replica[e]_4$};
    \path[<->,thin] (r11) edge (r12) edge (r13) edge (r14)
                    (r12) edge (r13) edge (r14)
                    (r13) edge (r14);
    \node[above,align=center] at (1, 2.25) {Cluster\\(European Data)\strut};

    \filldraw[very thick,fill=orange!10,draw=orange!10!black!30] (5.75, -0.25) rectangle (8.25, 2.25);
    \node (r21) at (6, 2) {$\Replica[a]_1$};
    \node (r22) at (8, 2) {$\Replica[a]_2$};
    \node (r23) at (6, 0) {$\Replica[a]_3$};
    \node (r24) at (8, 0) {$\Replica[a]_4$};
    \path[<->,thin] (r21) edge (r22) edge (r23) edge (r24)
                    (r22) edge (r23) edge (r24)
                    (r23) edge (r24);
    \node[above,align=center] at (7, 2.25) {Cluster\\(American Data)\strut};

    \path (2.3, 1.25) edge[double,<->] node[above] {Cluster Sending} node[below] {(coordination)} (5.7, 1.25);

    \node[below,align=center] (ec) at (1, -0.65) {Requests\\(European Data)\strut} edge[thick,->] (1,-0.25);
    \node[below,align=center] (ec) at (4, -0.65) {Requests\\(Mixed Data)\strut}    edge[thick,->] (4, -0.35);
    \node[below,align=center] (ac) at (7, -0.65) {Requests\\(American Data)\strut} edge[thick,->] (7,-0.25);
    \path[thick,black!75,->] (4, -0.35) edge[dotted,bend right=20] (2.3, 0.5) edge[dotted,bend left=20] (5.7, 0.5);
\end{tikzpicture}
\quad
\caption{On the \emph{left}, a traditional fully-replicated resilient system in which all four replicas each hold all data.  On the \emph{right}, a \emph{sharded} design in which each resilient cluster of four replicas holds only a part of the data.}\label{fig:example_intro}
\end{figure}
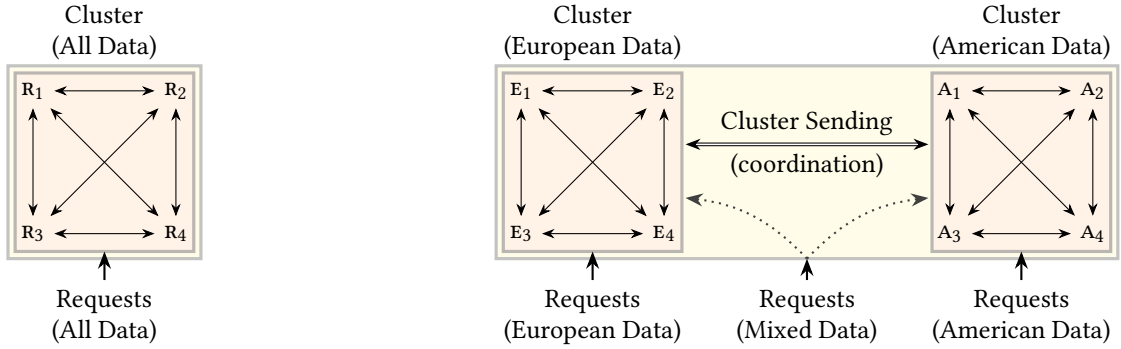

To support requests that affect data in several clusters in such a sharded design, the clusters need to be able to \emph{coordinate their operations}~\cite{chainspace,vldb}. Examples of such designs are \Name{AHL}~\cite{ahl}, \Name{ByShard}~\cite{byshard}, \Name{Chainspace}~\cite{chainspace}, and \Name{RingBFT}~\cite{ringbft}. One can base such coordination upon a \emph{cluster-sending protocol} that provides a basic Byzantine fault-tolerant communication primitive enabling communication between clusters~\cite{disc_csp}. Although cluster-sending has received some attention (e.g., as part of the design of \Name{GeoBFT}~\cite{vldb} and \Name{Chainspace}~\cite{chainspace}), and protocols with worst-case optimal complexity are known~\cite{disc_csp}, we believe there is still much room for improvement. In this paper, we introduce \emph{probabilistic cluster-sending} techniques that are able to provide low \emph{expected-case} message complexity (at the cost of higher communication latencies, a good trade-off in systems where inter-cluster network bandwidth is limited). In specific, our main contributions are as follows:
\begin{enumerate}
\item First, in Section~\ref{sec:cs_step}, we introduce the cluster-sending step \Name{cs-step} that attempts to send a value from a replica in the sending cluster to a replica in the receiving cluster in a verifiable manner and with a constant amount of inter-cluster communication.
\item Then, in Section~\ref{sec:random}, we introduce the \emph{Synchronous Probabilistic Cluster-Sending protocol} \Name{Pcs} that uses \Name{cs-step} with randomly selected sending and receiving replicas to provide cluster-sending in \emph{expected constant} steps. We also propose \emph{pruned-\Name{Pcs}} (\Name{Ppcs}), a fine-tuned version of \Name{Pcs} that guarantees termination.
\item In Section~\ref{sec:linear}, we propose the \emph{Synchronous Probabilistic Linear Cluster-Sending protocol} \Name{Plcs}, that uses \Name{cs-step} with a specialized randomized scheme to select replicas, this to provide cluster-sending in \emph{expected constant} steps and \emph{worst-case linear} steps, which is optimal.
\item Finally, in Section~\ref{sec:async}, we discuss how \Name{Pcs}, \Name{Ppcs}, and \Name{Plcs} can be generalized to operate in environments with \emph{asynchronous and unreliable communication}.
\end{enumerate}
A summary of our findings in comparison with existing techniques can be found in Table~\ref{tbl:summary}. In Section~\ref{sec:prelim}, we introduce the necessary terminology and notation, in Section~\ref{sec:related}, we compare with related work, and in Section~\ref{sec:concl}, we conclude on our findings. Finally, Appendix~\ref{app:first}--\ref{app:app_simply} provide complete proofs and other details not included in the main paper.

\begin{table}[t]
    \centering
    \caption{A comparison of \emph{cluster-sending protocols} that send a value from cluster $\Cluster_1$ with $\n{\Cluster_1}$ replicas, of which $\f{\Cluster_1}$ are faulty, to cluster $\Cluster_2$ with $\n{\Cluster_2}$ replicas, of which $\f{\Cluster_2}$ are faulty. For each protocol $P$, \emph{Protocol} specifies its name; \emph{Robustness} specifies the conditions $P$ puts on the clusters; \emph{Message Steps} specifies the number of messages exchanges $P$ performs; \emph{Optimal} specifies whether $P$ is worst-case optimal; and \emph{Unreliable} specifies whether $P$ can deal with unreliable communication.}\label{tbl:summary}
    \setlength\tabcolsep{2pt}
    \begin{minipage}{\textwidth}
    \small
    \renewcommand{\footnoterule}{\vskip-2pt}
    \centering
    \makebox[0pt]{
    \begin{tabular}{c|c|c|cc|cc}
        &Protocol&Robustness\footnote{Protocols that have different message step complexities depending on the robustness assumptions have been included for each of the robustness assumptions.}&\multicolumn{2}{c|}{Message Steps}&Optimal&Unreliable\\
        & & &(expected-case)&(worst-case)&\\
        \hline
        \hline
        &\Name{PBS-cs}~\cite{disc_csp}&\Mid{$\min(\n{\Cluster_1}, \n{\Cluster_2}) > \f{\Cluster_1} + \f{\Cluster_2}$}&\multicolumn{2}{c|}{\Mid{$\f{\Cluster_1} + \f{\Cluster_2}+1$}}&\CheckMark{}&\FailMark{}\\
        &\Name{PBS-cs}~\cite{disc_csp}&\Good{$\n{\Cluster_1} > 3\f{\Cluster_1}$, $\n{\Cluster_2} > 3\f{\Cluster_2}$}&\multicolumn{2}{c|}{\Mid{$\max(\n{\Cluster_1}, \n{\Cluster_2})$}}&\CheckMark{}&\FailMark{}\\
        \hline
        &\Name{GeoBFT}~\cite{vldb}&\Mid{$\n{\Cluster_1}=\n{\Cluster_2} > 3\max(\f{\Cluster_1}, \f{\Cluster_2})$}&\Mid{$\f{\Cluster_2} + 1$\footnote{Complexity when the coordinating primary in $\Cluster_1$ is non-faulty and communication is reliable.}}&\Bad{$\BigOm{\f{\Cluster_1}\n{\Cluster_2}}$}&\FailMark{}&\CheckMark{}\\
        \hline
        &\Name{Chainspace}~\cite{chainspace}&\Good{$\n{\Cluster_1} > 3\f{\Cluster_1}$, $\n{\Cluster_2} > 3\f{\Cluster_2}$}&\multicolumn{2}{c|}{\Bad{$\n{\Cluster_1}\n{\Cluster_2}$}}&\FailMark{}&\FailMark\\
        \hline
        \hline
        \multirow{5}{*}{\rotatebox[origin=c]{90}{This Paper}}
        &\Name{Ppcs}&\Good{$\n{\Cluster_1} > 2\f{\Cluster_1}$, $\n{\Cluster_2} > 2\f{\Cluster_2}$}&\Good{$4$}&\Bad{$(\f{\Cluster_1} + 1)(\f{\Cluster_2} + 1)$}&\FailMark{}&\CheckMark{}\\
        &\Name{Ppcs}&\Good{$\n{\Cluster_1} > 3\f{\Cluster_1}$, $\n{\Cluster_2} > 3\f{\Cluster_2}$}&\Good{$2\frac{1}{4}$}&\Bad{$(\f{\Cluster_1} + 1)(\f{\Cluster_2} + 1)$}&\FailMark{}&\CheckMark{}\\
        \hhline{~------}
        &\Name{Plcs}&\Mid{$\min(\n{\Cluster_1}, \n{\Cluster_2}) > \f{\Cluster_1} + \f{\Cluster_2}$}&\Good{$4$}&\Good{$\f{\Cluster_1} + \f{\Cluster_2}+1$}&\CheckMark{}&\CheckMark{}\\
        &\Name{Plcs}&\Mid{$\min(\n{\Cluster_1}, \n{\Cluster_2}) > 2(\f{\Cluster_1} + \f{\Cluster_2})$}&\Good{$2\frac{1}{4}$}&\Good{$\f{\Cluster_1} + \f{\Cluster_2}+1$}&\CheckMark{}&\CheckMark{}\\
        &\Name{Plcs}&\Good{$\n{\Cluster_1} > 3\f{\Cluster_1}$, $\n{\Cluster_2} > 3\f{\Cluster_2}$}&\Good{$3$}&\Good{$\max(\n{\Cluster_1}, \n{\Cluster_2})$}&\CheckMark{}&\CheckMark{}\\
        \hline
    \end{tabular}}
    \end{minipage}
\end{table}

\section{The Cluster-Sending Problem}\label{sec:prelim}
Before we present our probabilistic cluster-sending techniques, we first introduce all necessary terminology and notation. The formal model we use is based on the formalization of the cluster-sending problem provided by Hellings et al.~\cite{disc_csp}. If $S$ is a set of replicas, then $\Faulty{S} \subseteq S$ denotes the \emph{faulty replicas} in $S$, whereas $\NonFaulty{S} = S \difference \Faulty{S}$ denotes the \emph{non-faulty replicas} in $S$. We write $\n{S} = \abs{S}$, $\f{S} = \abs{\Faulty{S}}$, and $\nf{S} = \abs{\NonFaulty{S}} = \n{S}-\f{S}$ to denote the number of replicas, faulty replicas, and non-faulty replicas in $S$, respectively. A \emph{cluster} $\Cluster$ is a finite set of replicas. We consider clusters with \emph{Byzantine replicas} that behave in arbitrary manners. In specific, if $\Cluster$ is a cluster, then any malicious adversary can control the replicas in $\Faulty{\Cluster}$ at any time, but adversaries cannot bring non-faulty replicas under their control.
\begin{definition}
Let $\Cluster_1, \Cluster_2$ be disjoint clusters. The \emph{cluster-sending problem} is the problem of sending a value $v$ from $\Cluster_1$ to $\Cluster_2$ such that
\textnormal{\textbf{(1)}} all non-faulty replicas in $\NonFaulty{\Cluster_2}$ \Decide{receive} the value $v$;
\textnormal{\textbf{(2)}} all non-faulty replicas in $\NonFaulty{\Cluster_1}$ \Decide{confirm} that the value $v$ was received by all non-faulty replicas in $\NonFaulty{\Cluster_2}$; and
\textnormal{\textbf{(3)}} non-faulty replicas in $\NonFaulty{\Cluster_2}$ only receive a value $v$ if all non-faulty replicas in $\NonFaulty{\Cluster_1}$ \Decide{agree} upon sending $v$.
\end{definition}

We assume that there is no limitation on local communication within a cluster, while global communication between clusters is costly. This model is supported by practice, where communication between wide-area deployments of clusters is up-to-two orders of magnitudes more expensive than communication within a cluster~\cite{vldb}.

We assume that each cluster can make \emph{local decisions} among all non-faulty replicas, e.g., via a \emph{consensus protocol} such as \Name{Pbft} or \Name{Paxos}~\cite{pbftj,paxossimple}. Furthermore, we assume that the replicas in each cluster can certify such local decisions via a \emph{signature scheme}. E.g., a cluster $\Cluster$ can certify a consensus decision on some message $m$ by collecting a set of signatures for $m$ of $\f{\Cluster}+1$ replicas in $\Cluster$, guaranteeing one such signature is from a non-faulty replica (which would only signs values on which consensus is reached). We write $\Sign{m}{\Cluster}$ to denote a message $m$ certified by $\Cluster$. To minimize the size of certified messages, one can utilize a threshold signature scheme~\cite{rsasign}. To enable decision making and message certification, we assume, for every cluster $\Cluster$, $\n{\Cluster} > 2\f{\Cluster}$~\cite{netbound,dolevstrong,byzgen,consbound}. Lastly, we assume that there is a common source of randomness for all non-faulty replicas of each cluster, e.g., via a distributed fault-tolerant random coin~\cite{coin,bdrandcoin}. 

\section{The Cluster-Sending Step}\label{sec:cs_step}
 
If communication is reliable and one knows non-faulty replicas $\Replica_1 \in \NonFaulty{\Cluster_1}$ and $\Replica_2 \in \NonFaulty{\Cluster_2}$, then cluster-sending a value $v$ from $\Cluster_1$ to $\Cluster_2$ can be done via a straightforward \emph{cluster-sending step}. Under these conditions, one can simply instruct $\Replica_1$ to send $v$ to $\Replica_2$. When $\Replica_2$ receives $v$, it can disperse $v$ locally in $\Cluster_2$. Unfortunately, we do not know which replicas are faulty and which are non-faulty. Furthermore, it is practically impossible to reliably determine which replicas are non-faulty, as faulty replicas can appear well-behaved to most replicas, while interfering with the operations of only some non-faulty replicas.

To deal with faulty replicas when utilizing the above \emph{cluster-sending step}, one needs to build in sufficient safeguards to detect \emph{failure} of $\Replica_1$, of $\Replica_2$, or of the communication between them. To do so, we add receive and confirmation phases to the sketched cluster-sending step. During the \emph{receive phase}, the receiving replica $\Replica_2$ must construct a proof $P$ that it received and dispersed $v$ locally in $\Cluster_2$ and then send this proof back to $\Replica_1$. Finally, during the \emph{confirmation phase}, $\Replica_1$ can utilize $P$ to prove to all other replicas in $\Cluster_1$ that the cluster-sending step was successful. The pseudo-code of this \emph{cluster-sending step protocol} \Name{cs-step} can be found in Figure~\ref{fig:cs_step}. We have the following:

\begin{figure}[t]
    \begin{myprotocol}{\PROTOCALL{cs-step}{$\Replica_1$, $\Replica_2$, $v$}, with $\Replica_1 \in \Cluster_1$ and $\Replica_2 \in \Cluster_2$}
        \REQUIRE Each replica in $\NonFaulty{\Cluster_1}$ decided \Decide{agree} on sending $v$ to $\Cluster_2$ (and can construct $\SignMessage{send}{v,\ \Cluster_2}{\Cluster_1}$).
        \ENSURE 
            \textbf{(i)} If communication is reliable, $\Replica_1 \in \NonFaulty{\Cluster_1}$, and $\Replica_2 \in \NonFaulty{\Cluster_2}$, then $\Replica_1$ decides \Decide{confirm} on $v$.
            \textbf{(ii)} If a replica in $\NonFaulty{\Cluster_2}$ decides \Decide{receive} on $v$, then all replicas in $\NonFaulty{\Cluster_1}$ decided \Decide{agree} on sending $v$ to $\Cluster_2$.
            \textbf{(iii)} If a replica in $\NonFaulty{\Cluster_1}$ decides \Decide{confirm} on $v$, then all replicas in $\NonFaulty{\Cluster_2}$ decided \Decide{receive} on $v$  and all replicas in $\NonFaulty{\Cluster_1}$ eventually decide \Decide{confirm} on $v$ (whenever communication becomes reliable). 
        \SPACE
        \TITLE{The \emph{cluster-sending step} for $\Replica_1$ and $\Replica_2$}
        \STATE Instruct $\Replica_1$ to send $\SignMessage{send}{v,\ \Cluster_2}{\Cluster_1}$ to $\Replica_2$.\label{fig:cs_step:send_v}
        \SPACE
        \TITLE{The \emph{receive role} for $\Cluster_2$} 
        \EVENT{$\Replica_2 \in \NonFaulty{\Cluster_2}$ receives $m \GETS \SignMessage{send}{v,\ \Cluster_2}{\Cluster_1}$ from $\Replica_1 \in \Cluster_1$}\label{fig:cs_step:receive_event}
            \IF{$\Replica_2$ does not have consensus on $m$}
                \STATE Use \emph{local consensus} on $m$ and construct $\SignMessage{proof}{m}{\Cluster_ 2}$.\label{fig:cs_step:receive_v}
                \STATE \COMMENT{After local consensus, each replica in $\NonFaulty{\Cluster_2}$ decides \Decide{receive} on $v$.}\label{fig:cs_step:receive}
            \ENDIF
            \STATE Send $\SignMessage{proof}{m}{\Cluster_2}$ to $\Replica_1$.\label{fig:cs_step:send_confirm}
        \ENDEVENT
        \SPACE
        \TITLE{The \emph{confirmation role} for $\Cluster_1$}
            \EVENT{$\Replica_1 \in \NonFaulty{\Cluster_1}$ receives $m_p \GETS \SignMessage{proof}{\SignMessage{send}{v,\ \Cluster_2}{\Cluster_1}}{\Cluster_2}$ from $\Replica_2 \in \Cluster_2$}\label{fig:cs_step:replied_event}
            \IF{$\Replica_1$ does not have consensus on $m_p$}
                \STATE Use \emph{local consensus} on $m_p$.\label{fig:cs_step:confirm_v}
                \STATE \COMMENT{After local consensus, each replica in $\NonFaulty{\Cluster_1}$ decides \Decide{confirm} on $v$.}\label{fig:cs_step:confirm}
            \ENDIF
        \ENDEVENT
    \end{myprotocol}
    \caption{The Cluster-sending step protocol \PROTOCALL{cs-step}{$\Replica_1$, $\Replica_2$, $v$}. In this protocol, $\Replica_1$ tries to send $v$ to $\Replica_2$, which will succeed if both $\Replica_1$ and $\Replica_2$ are non-faulty.}\label{fig:cs_step}
\end{figure}

\begin{proposition}\label{prop:cs_step}
Let $\Cluster_1, \Cluster_2$ be disjoint clusters with $\Replica_1 \in \Cluster_1$ and $\Replica_2 \in \Cluster_2$. If $\Cluster_1$ satisfies the pre-conditions of \PROTOCALL{cs-step}{$\Replica_1$, $\Replica_2$, $v$}, then execution of \PROTOCALL{cs-step}{$\Replica_1$, $\Replica_2$, $v$} satisfies the post-conditions and will exchange at most two messages between $\Cluster_1$ and $\Cluster_2$.
\end{proposition}

In the following sections, we show how to use the cluster-sending step in the construction of cluster-sending protocols. In Section~\ref{sec:random}, we introduce synchronous protocols that provide \emph{expected constant message complexity}. Then, in Section~\ref{sec:linear}, we introduce synchronous protocols that additionally provide \emph{worst-case linear message complexity}, which is optimal. Finally, in Section~\ref{sec:async}, we show how to extend the presented techniques to asynchronous communication.

\section{Probabilistic Cluster-Sending with Random Replica Selection}\label{sec:random}

In the previous section, we introduced \Name{cs-step}, the cluster-sending step protocol that succeeds whenever the participating replicas are non-faulty and communication is reliable. Using \Name{cs-step}, we build a three-step  protocol that cluster-sends a value $v$ from $\Cluster_1$ to $\Cluster_2$:
\begin{enumerate}
\item First, the replicas in $\NonFaulty{\Cluster_1}$ reach agreement and decide \Decide{agree} on sending $v$ to $\Cluster_2$.
\item Then, the replicas in $\NonFaulty{\Cluster_1}$ perform a \emph{probabilistic cluster-sending step} by electing replicas $\Replica_1 \in \Cluster_1$ and $\Replica_2 \in \Cluster_2$ fully at random, after which \PROTOCALL{cs-step}{$\Replica_1$, $\Replica_2$, $v$} is executed.
\item Finally, each replicas in $\NonFaulty{\Cluster_1}$ waits for the completion of \PROTOCALL{cs-step}{$\Replica_1$, $\Replica_2$, $v$} If the waiting replicas decided \Decide{confirm} on $v$ during this wait, then cluster-sending is successful. Otherwise, we repeat the previous step.
\end{enumerate}
To enable replicas to wait for completion, we assume \emph{synchronous} inter-cluster communication: messages sent by non-faulty replicas will be delivered within some known bounded delay. Such \emph{synchronous} systems can be modeled by \emph{pulses}~\cite{min_sync,partsync}:

\begin{definition}
A system is \emph{synchronous} if all inter-cluster communication happens in \emph{pulses} such that every message sent in a pulse will be received in the same pulse.
\end{definition}

The pseudo-code of the resultant \emph{Synchronous Probabilistic Cluster-Sending protocol} \Name{Pcs} can be found in Figure~\ref{fig:cs}. Next, we prove that \Name{Pcs} performs cluster-sending with expected constant message complexity.

\begin{figure}
    \begin{myprotocol}{\PROTOCALL{Pcs}{$\Cluster_1$, $\Cluster_2$, $v$}}
        \STATE Use \emph{local consensus} on $v$ and construct $\SignMessage{send}{v,\ \Cluster_2}{\Cluster_1}$.\label{fig:cs:agree_v}
        \STATE \COMMENT{After local consensus, each replica in $\NonFaulty{\Cluster_1}$ decides \Decide{agree} on $v$.}\label{fig:cs:agree}
        \REPEAT\label{fig:cs:loop}
            \STATE Choose replicas $(\Replica_1, \Replica_2) \in \Cluster_1 \times \Cluster_2$, fully at random.\label{fig:cs:choose}
            \STATE \PROTOCALL{cs-step}{$\Replica_1$, $\Replica_2$, $v$}
            \STATE Wait \emph{three} global pulses.\label{fig:cs:wait}
        \UNTIL{$\Cluster_1$ reaches consensus on $\SignMessage{proof}{\SignMessage{send}{v,\ \Cluster_2}{\Cluster_1}}{\Cluster_2}$.}\label{fig:cs:loop_end}
    \end{myprotocol}
    \caption{The Synchronous Probabilistic Cluster-Sending protocol \PROTOCALL{Pcs}{$\Cluster_1$, $\Cluster_2$, $v$} that cluster-sends a value $v$ from $\Cluster_1$ to $\Cluster_2$.}\label{fig:cs}
\end{figure}

\begin{theorem}\label{thm:cs}
Let $\Cluster_1, \Cluster_2$ be disjoint clusters. If communication is synchronous, then execution of \PROTOCALL{Pcs}{$\Cluster_1$, $\Cluster_2$, $v$} results in cluster-sending $v$ from $\Cluster_1$ to $\Cluster_2$. The execution performs two local consensus steps in $\Cluster_1$, one local consensus step in $\Cluster_2$, and is expected to make $\dsfrac{(\n{\Cluster_1}\n{\Cluster_2})}{(\nf{\Cluster_1}\nf{\Cluster_1})}$ cluster-sending steps.
\end{theorem}
\begin{proof}
Due to \lsfref{fig:cs}{agree_v}{agree}, \PROTOCALL{Pcs}{$\Cluster_1$, $\Cluster_2$, $v$} establishes the pre-conditions for any execution of \Name{cs-step}($\Replica_1$, $\Replica_2$, $v$) with $\Replica_1 \in \Cluster_1$ and $\Replica_2 \in \Cluster_2$. Using the correctness of \Name{cs-step} (Proposition~\ref{prop:cs_step}), we conclude that \PROTOCALL{Pcs}{$\Cluster_1$, $\Cluster_2$, $v$} results in cluster-sending $v$ from $\Cluster_1$ to $\Cluster_2$ whenever the replicas $(\Replica_1, \Replica_2) \in \Cluster_1 \times \Cluster_2$ chosen at \lfref{fig:cs}{choose} are non-faulty. As the replicas $(\Replica_1, \Replica_2) \in \Cluster_1 \times \Cluster_2$ are chosen fully at random, we have probability $p_i = \dsfrac{\nf{\Cluster_i}}{\n{\Cluster_i}}$, $i \in \{1,2\}$, of choosing $\Replica_i \in \NonFaulty{\Cluster_i}$. The probabilities $p_1$ and $p_2$ are independent of each other. Consequently, the probability on choosing $(\Replica_1, \Replica_2) \in \NonFaulty{\Cluster_1} \times \NonFaulty{\Cluster_2}$ is $p = p_1p_2 = \dsfrac{(\nf{\Cluster_1}\nf{\Cluster_2})}{(\n{\Cluster_1}\n{\Cluster_2})}$. As such, each iteration of the loop at \lfref{fig:cs}{loop} can be modeled as an independent \emph{Bernoulli trial} with probability of success $p$, and the expected number of iterations of the loop is $p^{-1} = \dsfrac{(\n{\Cluster_1}\n{\Cluster_2})}{(\nf{\Cluster_1}\nf{\Cluster_1})}$.

Finally, we prove that each local consensus step needs to be performed only once. To do so, we consider the local consensus steps triggered by the loop at \lfref{fig:cs}{loop}. These are the local consensus steps at Lines~\ref{fig:cs_step:receive_v} and~\ref{fig:cs_step:confirm_v} of Figure~\ref{fig:cs_step}. The local consensus step at Line~\ref{fig:cs_step:receive_v} can be initiated by a faulty replica $\Replica_2$. After this single local consensus step reaches consensus on message $m \GETS \SignMessage{send}{v,\ \Cluster_2}{\Cluster_1}$, each replica in $\NonFaulty{\Cluster_2}$ reaches consensus on $m$, decides \Decide{receive} on $v$, and can construct $m_p \GETS \SignMessage{proof}{m}{\Cluster_2}$, this independent of the behavior of $\Replica_2$. Hence, a single local consensus step for $m$ in $\Cluster_2$ suffices, and no replica in $\NonFaulty{\Cluster_2}$ will participate in future consensus steps for $m$. An analogous argument proves that a single local consensus step for $m_p$ in $\Cluster_1$, performed at \lfref{fig:cs_step}{confirm_v}, suffices.
\end{proof}

In typical fault-tolerant clusters, at least half of the replicas are non-faulty (e.g., in synchronous systems with Byzantine failures that use digital signatures, or in systems that only deal with crashes) or at least two-third of the replicas are non-faulty (e.g., asynchronous systems). In these systems, \Name{Pcs} is expected to only performs a few cluster-sending steps:

\begin{corollary}\label{cor:cs}
Let $\Cluster_1, \Cluster_2$ be disjoint clusters. If communication is synchronous, then the expected number of cluster-sending steps performed by \PROTOCALL{Pcs}{$\Cluster_1$, $\Cluster_2$, $v$} is upper bounded by $4$ if $\n{\Cluster_1} > 2\f{\Cluster_1}$ and $\n{\Cluster_2} > 2\f{\Cluster_2}$; and by $2\frac{1}{4}$ if $\n{\Cluster_1} > 3\f{\Cluster_1}$ and $\n{\Cluster_2} > 3\f{\Cluster_2}$.
\end{corollary}

In \Name{Pcs}, the replicas $(\Replica_1, \Replica_2) \in \Cluster_1 \times \Cluster_2$ are chosen fully at random and \emph{with replacement}, as \Name{Pcs} does not retain any information on \emph{failed} probabilistic steps. In the worst case, this prevents \emph{termination}, as the same pair of replicas can be picked repeatedly. Furthermore, \Name{Pcs} does not prevent the choice of faulty replicas whose failure could be detected. We can easily improve on this, as the \emph{failure} of a probabilistic step provides some information on the chosen replicas. In specific, we have the following technical properties:

\begin{lemma}\label{lem:cs_prune}
Let $\Cluster_1, \Cluster_2$ be disjoint clusters. We assume synchronous communication and assume that each replica in $\NonFaulty{\Cluster_1}$ decided \Decide{agree} on sending $v$ to $\Cluster_2$.
\begin{enumerate}
    \item\label{lem:cs_prune:norep} Let $(\Replica_1, \Replica_2) \in \Cluster_1 \times \Cluster_2$. If \PROTOCALL{cs-step}{$\Replica_1$, $\Replica_2$, $v$} fails to cluster-send $v$, then either $\Replica_1 \in \Faulty{\Cluster_1}$, $\Replica_2 \in \Cluster_2$, or both.
    \item\label{lem:cs_prune:c1f} Let $\Replica_1 \in \Cluster_1$. If \PROTOCALL{cs-step}{$\Replica_1$, $\Replica_2$, $v$} fails to cluster-send $v$ for $\f{\Cluster_2} + 1$ distinct replicas $\Replica_2 \in \Cluster_2$, then $\Replica_1 \in \Faulty{\Cluster_1}$.
    \item\label{lem:cs_prune:c2f} Let $\Replica_2 \in \Cluster_2$. If \PROTOCALL{cs-step}{$\Replica_1$, $\Replica_2$, $v$} fails to cluster-send $v$ for $\f{\Cluster_1} + 1$ distinct replicas $\Replica_1 \in \Cluster_1$, then $\Replica_2 \in \Faulty{\Cluster_2}$.
\end{enumerate}
\end{lemma}

We can apply the properties of Lemma~\ref{lem:cs_prune} to actively \emph{prune} which replica pairs \Name{Pcs} considers (\lfref{fig:cs}{choose}). Notice that pruning via Lemma~\subref{lem:cs_prune}{norep} simply replaces choosing replica pairs \emph{with replacement}, as done by \Name{Pcs}, by choosing replica pairs \emph{without replacement}, this without further reducing the possible search space. Indeed, if we \emph{only} apply this pruning step, then, in the worst case, we still need $\f{\Cluster_1}\nf{\Cluster_2} + \nf{\Cluster_1}\f{\Cluster_2} + \f{\Cluster_1}\f{\Cluster_2} + 1$ cluster-sending steps. Pruning via Lemma~\subref{lem:cs_prune}{c1f} does reduce the search space, however, as each replica in $\Cluster_1$ will only be paired with a subset of $\f{\Cluster_2} + 1$ replicas in $\Cluster_2$. Likewise, pruning via Lemma~\subref{lem:cs_prune}{c2f} also reduces the search space. We obtain the \emph{Pruned Synchronous Probabilistic Cluster-Sending protocol} (\Name{Ppcs}) by applying all three prune steps to \Name{Pcs}. By construction, Theorem~\ref{thm:cs}, and Lemma~\ref{lem:cs_prune}, we conclude:

\begin{corollary}\label{cor:pcs}
Let $\Cluster_1, \Cluster_2$ be disjoint clusters. If communication is synchronous, then execution of \PROTOCALL{Ppcs}{$\Cluster_1$, $\Cluster_2$, $v$} results in cluster-sending $v$ from $\Cluster_1$ to $\Cluster_2$. The execution performs two local consensus steps in $\Cluster_1$, one local consensus step in $\Cluster_2$, is expected to make less than $\dsfrac{(\n{\Cluster_1}\n{\Cluster_2})}{(\nf{\Cluster_1}\nf{\Cluster_1})}$ cluster-sending steps, and makes worst-case $(\f{\Cluster_1}+1)(\f{\Cluster_2} +1)$ cluster-sending steps.
\end{corollary}

\section{Worst-Case Linear-Time Probabilistic Cluster-Sending}\label{sec:linear}
In the previous section, we introduced \Name{Pcs} and \Name{Ppcs}, two probabilistic cluster-sending protocols that can cluster-send a value $v$ from $\Cluster_1$ to $\Cluster_2$ with expected constant cost. Unfortunately, \Name{Pcs} does not guarantee termination, while \Name{Ppcs} has a worst-case \emph{quadratic complexity}. To improve on this, we need to improve the scheme by which we select replica pairs $(\Replica_1, \Replica_2) \in \Cluster_1 \times \Cluster_2$ that we use in cluster-sending steps. The straightforward manner to guarantee a worst-case \emph{linear complexity} is by using a scheme that can select only up-to-$n = \max(\n{\Cluster_1}, \n{\Cluster_2})$ distinct pairs $(\Replica_1, \Replica_2) \in \Cluster_1 \times \Cluster_2$. To select $n$ replica pairs from $\Cluster_1 \times \Cluster_2$, we will proceed in two steps.
\begin{enumerate}
    \item We generate list $S_1$ of $n$ replicas taken from $\Cluster_1$ and list $S_2$ of $n$ replicas taken from $\Cluster_2$.
    \item Then, we choose permutations $P_1 \in \Permute{S_1}$ and $P_2 \in \Permute{S_2}$ fully at random, and interpret each pair $(P_1[i], P_2[i])$. $0 \leq i < n$, as one of the chosen replica pairs.
\end{enumerate}
We use the first step to deal with any differences in the sizes of $\Cluster_1$ and $\Cluster_2$, and we use the second step to introduce sufficient randomness in our protocol.

Next, we introduce some notations to simplify reasoning about the above list-based scheme. If $R$ is a set of replicas, then $\List{R}$ is the list consisting of the replicas in $R$ placed in a predetermined order (e.g., on increasing replica identifier). If $S$ is a list of replicas, then we write $\Faulty{S}$ to denote the faulty replicas in $S$ and $\NonFaulty{S}$ to denote the non-faulty replicas in $S$, and we write $\n{S} = \abs{S}$, $\f{S} = \abs{\{ i \mid (0 \leq i < \n{S}) \land S[i] \in \Faulty{S}\} }$, and $\nf{S} = \n{S} - \f{S}$ to denote the number of positions in $S$ with replicas, faulty replicas, and non-faulty replicas, respectively. If $(P_1, P_2)$ is a pair of equal-length lists of $n = \abs{P_1} = \abs{P_2}$ replicas, then we say that the $i$-th position is a \emph{faulty position} if either $P_1[i] \in \Faulty{P_1}$ or $P_2[i] \in \Faulty{P_2}$. We write $\FaultyPos{P_1}{P_2}$ to denote the number of \emph{faulty positions} in $(P_1, P_2)$. As faulty positions can only be constructed out of the $\f{P_1}$ faulty replicas in $P_1$ and the $\f{P_2}$ faulty replicas in $P_2$, we must have $\max(\f{P_1}, \f{P_2}) \leq \FaultyPos{P_1}{P_2} \leq \min(n, \f{P_1} + \f{P_2})$,

\begin{example}\label{ex:list_pairs}
Consider clusters $\Cluster_1, \Cluster_2$ with $S_1 = \List{\Cluster_1} = [ \Replica_{1,1}, \dots, \Replica_{1,5}]$, $\Faulty{\Cluster_1} = \{ \Replica_{1,1}, \Replica_{1,2} \}$, $S_2 = \List{\Cluster_2} = [ \Replica_{2,1}, \dots, \Replica_{2, 5}]$, and $\Faulty{\Cluster_2} = \{\Replica_{2,1}, \Replica_{2,2} \}$. The set $\Permute{S_1} \times \Permute{S_2}$ contains $5!^2 = 14400$ list pairs. Now, consider the list pairs $(P_1, P_2), \allowbreak (Q_1, Q_2), \allowbreak (R_1, R_2) \in \Permute{S_1} \times \Permute{S_2}$ with
\[
    \begin{array}{r@{{} = [}c@{,}c@{,}c@{,}c@{,}c@{]}@{,\qquad}r@{{} = [}c@{,}c@{,}c@{,}c@{,}c@{]}}
        P_1&\BadT{\Replica_{1,1}}&\Replica_{1,5}&\BadT{\Replica_{1,2}}&\Replica_{1,4}&\Replica_{1,3}&
        P_2&\BadT{\Replica_{2,1}}&\Replica_{2,3}&\BadT{\Replica_{2,2}}&\Replica_{2,5}&\Replica_{2,4}\\
        Q_1&\BadT{\Replica_{1,1}}&\Replica_{1,3}&\Replica_{1,5}&\Replica_{1,4}&\BadT{\Replica_{1,2}}&
        Q_2&\Replica_{2,5}&\Replica_{2,4}&\Replica_{2,3}&\BadT{\Replica_{2,2}}&\BadT{\Replica_{2,1}}\\
        R_1&\Replica_{1,5}&\Replica_{1,4}&\Replica_{1,3}&\BadT{\Replica_{1,2}}&\BadT{\Replica_{1,1}}&
        R_2&\BadT{\Replica_{2,1}}&\BadT{\Replica_{2,2}}&\Replica_{2,3}&\Replica_{2,4}&\Replica_{2,5}.
    \end{array}
\]
We have underlined the faulty replicas in each list, and we have $\FaultyPos{P_1}{P_2} = 2 = \f{S_1} = \f{S_2}$, $\FaultyPos{Q_1}{Q_2} = 3$, and $\FaultyPos{R_1}{R_2} = 4 = \f{S_1} + \f{S_2}$.
\end{example}

In the following, we will use a \emph{list-pair function} $\SF$ to compute the initial list-pair $(S_1, S_2)$ of $n$ replicas taken from $\Cluster_1$ and $\Cluster_2$, respectively. Next, we build a cluster-sending protocol that uses $\SF$ to compute $S_1$ and $S_2$, uses randomization to choose $n$ replica pairs from $S_1 \times S_2$, and, finally, performs cluster-sending steps using only these $n$ replica pairs. The pseudo-code of the resultant \emph{Synchronous Probabilistic Linear Cluster-Sending protocol} \Name{Plcs} can be found in Figure~\ref{fig:plcs}. Next, we prove that \Name{Plcs} performs cluster-sending with a worst-case linear number of cluster-sending steps.

\begin{figure}
    \begin{myprotocol}{\PROTOCALL{Plcs}{$\Cluster_1$, $\Cluster_2$, $v$, $\SF$}}
        \STATE Use \emph{local consensus} on $v$ and construct $\SignMessage{send}{v,\ \Cluster_2}{\Cluster_1}$.\label{fig:plcs:agree_v}
        \STATE \COMMENT{After reaching local consensus, each replica in $\NonFaulty{\Cluster_1}$ decides \Decide{agree} on $v$.}\label{fig:plcs:agree}
        \STATE Let $(S_1, S_2) \GETS \SF(\Cluster_1, \Cluster_2)$.
        \STATE Choose $(P_1, P_2) \in \Permute{S_1} \times \Permute{S_2}$ fully at random.\label{fig:plcs:permute}
        \STATE $i \GETS 0$.
        \REPEAT\label{fig:plcs:loop}
            \STATE \PROTOCALL{cs-step}{$P_1[i]$, $P_2[i]$, $v$}\label{fig:plcs:choose}
            \STATE Wait \emph{three} global pulses.\label{fig:plcs:wait}
            \STATE $i \GETS i + 1$.
        \UNTIL{$\Cluster_1$ reaches consensus on $\SignMessage{proof}{\SignMessage{send}{v,\ \Cluster_2}{\Cluster_1}}{\Cluster_2}$.}\label{fig:plcs:loop_end}
    \end{myprotocol}
    \caption{The Synchronous Probabilistic Linear Cluster-Sending protocol \PROTOCALL{Plcs}{$\Cluster_1$, $\Cluster_2$, $v$, $\SF$} that cluster-sends a value $v$ from $\Cluster_1$ to $\Cluster_2$ using list-pair function $\SF$.}\label{fig:plcs}
\end{figure}

\begin{proposition}\label{prop:lpcs}
Let $\Cluster_1, \Cluster_2$ be disjoint clusters and let $\SF$ be a list-pair function with $(S_1, S_2) \GETS \SF(\Cluster_1, \Cluster_2)$ and $n = \n{S_1} = \n{S_2}$. If communication is synchronous and $n > \f{S_1} + \f{S_2}$, then execution of \PROTOCALL{Plcs}{$\Cluster_1$, $\Cluster_2$, $v$, $\SF$} results in cluster-sending $v$ from $\Cluster_1$ to $\Cluster_2$. The execution performs two local consensus steps in $\Cluster_1$, one local consensus step in $\Cluster_2$, and makes worst-case $\f{S_1} + \f{S_2} + 1$ cluster-sending steps.
\end{proposition}
\begin{proof}
Due to \lsfref{fig:plcs}{agree_v}{agree}, \PROTOCALL{Plcs}{$\Cluster_1$, $\Cluster_2$, $v$, $\SF$} establishes the pre-conditions for any execution of \Name{cs-step}($\Replica_1$, $\Replica_2$, $v$) with $\Replica_1 \in \Cluster_1$ and $\Replica_2 \in \Cluster_2$. Now let $(P_1, P_2) \in \Permute{S_1} \times \Permute{S_2}$, as chosen at \lfref{fig:plcs}{permute}. As $P_i$, $i \in \{1,2\}$, is a permutation of $S_i$, we have $\f{P_i} = \f{S_i}$. Hence, we have $\FaultyPos{P_1}{P_2} \leq \f{S_1} + \f{S_2}$ and there must exist a position $j$, $0 \leq j < n$, such that $(P_1[j], P_2[j]) \in \NonFaulty{\Cluster_1} \times \NonFaulty{\Cluster_2}$. Using the correctness of \Name{cs-step} (Proposition~\ref{prop:cs_step}), we conclude that \PROTOCALL{Plcs}{$\Cluster_1$, $\Cluster_2$, $v$, $\SF$} results in cluster-sending $v$ from $\Cluster_1$ to $\Cluster_2$ in at most $\f{S_1} + \f{S_2} + 1$ cluster-sending steps. Finally, the bounds on the number of consensus steps follow from an argument analogous to the one in the proof of Theorem~\ref{thm:cs}.
\end{proof}

Next, we proceed in two steps to arrive at practical instances of \Name{Plcs} with expected constant message complexity. First, in Section~\ref{ss:nfptp}, we study the probabilistic nature of \Name{Plcs}. Then, in Section~\ref{ss:complex_lcs}, we propose practical list-pair functions and show that these functions yield instances of \Name{Plcs} with expected constant message complexity.

\subsection{On the Expected-Case Complexity of \Name{Plcs}}\label{ss:nfptp}

As the first step to determine the expected-case complexity of \Name{Plcs}, we solve the following abstract problem that captures the probabilistic argument at the core of the expected-case complexity of \Name{Plcs}:
\begin{problem}[non-faulty position trials]
Let $S_1$ and $S_2$ be lists of $\abs{S_1} = \abs{S_2} = n$ replicas. Choose permutations $(P_1, P_2) \in \Permute{S_1} \times \Permute{S_2}$ fully at random. Next, we inspect positions in $P_1$ and $P_2$ fully at random (with replacement). The \emph{non-faulty position trials problem} asks how many positions one expects to inspect to find the first non-faulty position.
\end{problem}

Let $S_1$ and $S_2$ be list of $\abs{S_1} = \abs{S_2} = n$ replicas. To answer the non-faulty position trials problem, we first take a look into the combinatorics of \emph{faulty positions} in pairs $(P_1, P_2) \in \Permute{S_1} \times \Permute{S_2}$. Let $m_1 = \f{S_1}$ and $m_2 = \f{S_2}$. By $\FC{n}{m_1}{m_2}{k}$, we denote the number of distinct pairs $(P_1, P_2) \in \Permute{S_1} \times \Permute{S_2}$ one can construct that have exactly $k$ faulty positions, hence, with $\FaultyPos{P_1}{P_2} = k$. As observed, we have $\max(m_1, m_2) \leq \FaultyPos{P_1}{P_2} \leq \min(n, m_1 + m_2)$ for any pair $(P_1, P_2) \in \Permute{S_1} \times \Permute{S_2}$. Hence, we have $\FC{n}{m_1}{m_2}{k} = 0$ for all $k < \max(m_1, m_2)$ and $k > \min(n, m_1 + m_2)$.

Now consider the step-wise construction of any permutation $(P_1, P_2) \in \Permute{S_1} \times \Permute{S_2}$ with $k$ faulty positions. First, we choose $(P_1[0], P_2[0])$, the pair at position $0$, after which we choose pairs for the remaining $n-1$ positions. For $P_i[0]$, $i \in \{1,2\}$, we can choose $n$ distinct replicas, of which $m_i$ are faulty. If we pick a non-faulty replica, then the remainder of $P_i$ is constructed out of $n-1$ replicas, of which $m_i$ are faulty. Otherwise, the remainder of $P_i$ is constructed out of $n-1$ replicas of which $m_i - 1$ are faulty. If, due to our choice of $(P_1[0], P_2[0])$, the first position is faulty, then only $k-1$ out of the $n-1$ remaining positions must be faulty. Otherwise, $k$ out of the $n-1$ remaining positions must be faulty. Combining this analysis yields four types for the first pair $(P_1[0], P_2[0])$:
\begin{enumerate}
    \item \emph{A non-faulty pair} $(P_1[0], P_2[0]) \in \NonFaulty{P_1} \times \NonFaulty{P_2}$. We have $(n - m_1)(n-m_2)$ such pairs, and we have $\FC{n-1}{m_1}{m_2}{k}$ different ways to construct the remainder of $P_1$ and $P_2$.
    \item \emph{A $1$-faulty pair} $(P_1[0], P_2[0]) \in \Faulty{P_1} \times \NonFaulty{P_2}$. We have $m_1(n-m_2)$ such pairs, and we have $\FC{n-1}{m_1-1}{m_2}{k-1}$ different ways to construct the remainder of $P_1$ and $P_2$.
    \item \emph{A $2$-faulty pair} $(P_1[0], P_2[0]) \in \NonFaulty{P_1} \times \Faulty{P_2}$. We have $(n - m_1)m_2$ such pairs, and we have $\FC{n-1}{m_1}{m_2-2}{k-1}$ different ways to construct the remainder of $P_1$ and $P_2$.
    \item \emph{A both-faulty pair} $(P_1[0], P_2[0]) \in \Faulty{P_1} \times \Faulty{P_2}$. We have $m_1m_2$ such pairs, and we have $\FC{n-1}{m_1-1}{m_2-1}{k-1}$ different ways to construct the remainder of $P_1$ and $P_2$.
\end{enumerate} 
Hence, $\FC{n}{m_1}{m_2}{k}$ with $\max(m_1, m_2) \leq k \leq \min(n, m_1 + m_2)$ is recursively defined by:
\begin{align*}
\FC{n}{m_1}{m_2}{k}
    ={}&  (n-m_1) (n-m_2) \FC{n-1}{m_1}{m_2}{k}  &&\text{(non-faulty pair)}\\
       &+ m_1 (n-m_2) \FC{n-1}{m_1-1}{m_2}{k-1} &&\text{($1$-faulty pair)}\\
       &+ (n-m_1) m_2 \FC{n-1}{m_1}{m_2-1}{k-1} &&\text{($2$-faulty pair)}\\
       &+ m_1 m_2 \FC{n-1}{m_1-1}{m_2-1}{k-1},  &&\text{(both-faulty pair)}
\end{align*}
and the base case for this recursion is $\FC{0}{0}{0}{0}= 1$.

\begin{example}
Reconsider the list pairs $(P_1, P_2)$, $(Q_1, Q_2)$, and $(R_1, R_2)$ from Example~\ref{ex:list_pairs}:
\[
    \begin{array}{r@{{} = [}c@{,}c@{,}c@{,}c@{,}c@{]}@{,\qquad}r@{{} = [}c@{,}c@{,}c@{,}c@{,}c@{]}}
        P_1&\BadT{\Replica_{1,1}}&\Replica_{1,5}&\BadT{\Replica_{1,2}}&\Replica_{1,4}&\Replica_{1,3}&
        P_2&\BadT{\Replica_{2,1}}&\Replica_{2,3}&\BadT{\Replica_{2,2}}&\Replica_{2,5}&\Replica_{2,4}\\
        Q_1&\BadT{\Replica_{1,1}}&\Replica_{1,3}&\Replica_{1,5}&\Replica_{1,4}&\BadT{\Replica_{1,2}}&
        Q_2&\Replica_{2,5}&\Replica_{2,4}&\Replica_{2,3}&\BadT{\Replica_{2,2}}&\BadT{\Replica_{2,1}}\\
        R_1&\Replica_{1,5}&\Replica_{1,4}&\Replica_{1,3}&\BadT{\Replica_{1,2}}&\BadT{\Replica_{1,1}}&
        R_2&\BadT{\Replica_{2,1}}&\BadT{\Replica_{2,2}}&\Replica_{2,3}&\Replica_{2,4}&\Replica_{2,5}.
    \end{array}
\]
Again, we have underlined the faulty replicas in each list. In $(P_1, P_2)$, we have both-faulty pairs at positions $0$ and $2$ and non-faulty pairs at positions $1$, $3$, and $4$. In $(Q_1, Q_2)$, we have a $1$-faulty pair at position $0$, non-faulty pairs at positions $1$ and $2$, a $2$-faulty pair at position $3$, and a both-faulty pair at position $4$. Finally, in $(R_1, R_2)$, we have $2$-faulty pairs at positions $0$ and $1$, a non-faulty pair at position $2$, and $1$-faulty pairs at positions $3$ and $4$.
\end{example}

Using the above combinatorics of faulty positions, we can formalize an exact solution to the \emph{non-faulty position trials problem}:

\begin{lemma}\label{lem:gpt_prob_open}
Let $S_1$ and $S_2$ be lists of $n = \n{S_1} = \n{S_2}$ replicas with $m_1 = \f{S_1}$ and $m_2 = \f{S_2}$. If $m_1 + m_2 < n$, then the non-faulty position trials problem has solution
\[ \PT{n}{m_1}{m_2} = \frac{1}{n!^2} \left(\sum_{k = \max(m_1, m_2)}^{m_1 + m_2}\ \frac{n}{n-k} \FC{n}{m_1}{m_2}{k} \right). \]
\end{lemma}
\begin{proof}
We have $\abs{\Permute{S_1}} = \abs{\Permute{S_2}} = n!$. Consequently, we have $\abs{\Permute{S_1} \times \Permute{S_2}} = n!^2$ and we have probability $\dsfrac{1}{(n!^2)}$ to choose any pair $(P_1, P_2) \in \Permute{S_1} \times \Permute{S_2}$. Now consider such a pair $(P_1, P_2) \in \Permute{S_1} \times \Permute{S_2}$. As there are $\FaultyPos{P_1}{P_2}$ faulty positions in $(P_1, P_2)$, we have probability $p(P_1, P_2) = \dsfrac{(n - \FaultyPos{P_1}{P_2})}{n}$ to inspect a non-faulty position. Notice that $\max(m_1, m_2) \leq \FaultyPos{P_1}{P_2} \leq m_1 + m_2 < n$ and, hence, $0 < p(P_1, P_2) \leq 1$. Each of the inspected positions in $(P_1, P_2)$ is chosen fully at random. Hence, each inspection is a \emph{Bernoulli trial} with probability of success $p(P_1, P_2)$, and we expect to inspect a first non-faulty position in the $p(P_1, P_2)^{-1} = \dsfrac{n}{(n - \FaultyPos{P_1}{P_2})}$-th attempt. We conclude
\[\PT{n}{m_1}{m_2} = \frac{1}{n!^2} \left(\sum_{(P_1, P_2) \in \Permute{S_1} \times \Permute{S_2}}\ \frac{n}{n - \FaultyPos{P_1}{P_2}} \right).\]
Notice that there are $\FC{n}{m_1}{m_2}{k}$ distinct pairs $(P_1, P_2) \in \Permute{S_1} \times \Permute{S_2}$ with $\FaultyPos{P_1'}{P_2'} = k$ for each $k$, $\max(m_1, m_2) \leq k \leq m_1 + m_2 < n$. Hence, in the above expression for $\PT{n}{m_1}{m_2}$, we can group on these pairs $(P_1', P_2')$ to obtain the searched-for solution.
\end{proof}

To further solve the non-faulty position trials problem, we work towards a \emph{closed form} for $\FC{n}{m_1}{m_2}{k}$. Consider any pair $(P_1, P_2) \in \Permute{S_1} \times \Permute{S_2}$ with $\FaultyPos{P_1}{P_2} = k$ obtained via the outlined step-wise construction. Let $b_1$ be the number of \emph{$1$-faulty pairs}, let $b_2$ be the number of \emph{$2$-faulty pairs}, and let $b_{1,2}$ be the number of \emph{both-faulty pairs} in $(P_1, P_2)$. By construction, we must have
     $k = b_1 + b_2 + b_{1,2}$, $m_1 = b_1 + b_{1,2}$, and $m_2 = b_2 + b_{1,2}$ and by rearranging terms, we can derive
\begin{align*}
b_{1,2} &= (m_1 + m_2) - k,\\
    b_1 &= k - m_2,\\
    b_2 &= k - m_1.
\end{align*}

\begin{example}\label{ex:bbb}
Consider $S_1 = [\Replica_{1,1}, \dots, \Replica_{1,5}]$ with $\Faulty{S_1} = \{ \Replica_{1,1}, \Replica_{1,2}, \Replica_{1,3} \}$ and $S_2 = [\Replica_{2,1}, \dots, \Replica_{2,5} ]$ with $\Faulty{S_2} = \{\Replica_{2,1}\}$. Hence, we have $n = 5$, $m_1 = \f{S_1} = 3$, and $m_2 = \f{S_2} = 1$. If we want to create a pair $(P_1, P_2) \in \Permute{S_1} \times \Permute{S_2}$ with $k = \FaultyPos{P_1}{P_2} = 3$ faulty positions, then $(P_1, P_2)$ must have two non-faulty pairs, two $1$-faulty pairs, no $2$-faulty pairs, and one both-faulty pair. Hence, we have $n - k = 2$, $b_1 = 2$, $b_2 = 0$, and $b_{1,2} = 1$.
\end{example}

The above analysis only depends on the choice of $m_1$, $m_2$, and $k$, and not on our choice of $(P_1, P_2)$. Next, we use this analysis to express $\FC{n}{m_1}{m_2}{k}$ in terms of the number of distinct ways in which one can \emph{construct} 
\begin{enumerate}
    \renewcommand{\theenumi}{\Alph{enumi}}
    \renewcommand{\labelenumi}{(\theenumi)}
    \item\label{lt:b1} lists of $b_1$ $1$-faulty pairs out of faulty replicas from $S_1$ and non-faulty replicas from $S_2$,
    \item\label{lt:b2} lists of $b_2$ $2$-faulty pairs out of non-faulty replicas from $S_1$ and faulty replicas from $S_2$,
    \item\label{lt:b12} lists of $b_{1,2}$ both-faulty pairs out of the remaining faulty replicas in $S_1$ and $S_2$ that are not used in the previous two cases, and
    \item\label{lt:nk} lists of $n-k$ non-faulty pairs out of the remaining (non-faulty) replicas in $S_1$ and $S_2$ that are not used in the previous three cases;
\end{enumerate}
and in terms of the number of distinct ways one can \emph{merge} these lists. As the first step, we look at how many distinct ways we can merge two lists together:

\begin{lemma}\label{lem:list_merge}
For any two disjoint lists $S$ and $T$ with $\abs{S} = v$ and $\abs{T} = w$, there exist $\LM{v}{w} = \dsfrac{(v + w)!}{(v!w!)}$ distinct lists $L$ with  $\Limit{L}{S} = S$ and $\Limit{L}{T} = T$, in which $\Limit{L}{M}$, $M \in \{S,T\}$, is the list obtained from $L$ by only keeping the values that also appear in list $M$.
\end{lemma}

Next, we look at the number of distinct ways in which one can construct lists of type~\ref{lt:b1},~\ref{lt:b2},~\ref{lt:b12}, and~\ref{lt:nk}. Consider the construction of a list of type~\ref{lt:b1}. We can choose $\binom{m_1}{b_1}$ distinct sets of $b_1$ faulty replicas from $S_1$ and we can choose $\binom{n - m_2}{b_1}$ distinct sets of $b_1$ non-faulty replicas from $S_2$. As we can order the chosen values from $S_1$ and $S_2$ in $b_1!$ distinct ways, we can construct $b_1!^2\binom{m_1}{b_1} \binom{n-m_2}{b_1}$ distinct lists of type~\ref{lt:b1}. Likewise, we can construct $b_2!^2\binom{n-m_1}{b_2} \binom{m_2}{b_2}$ distinct lists of type~\ref{lt:b2}.

\begin{example}
We continue from the setting of Example~\ref{ex:bbb}: we want to create a pair $(P_1, P_2) \in \Permute{S_1} \times \Permute{S_2}$ with $k = \FaultyPos{P_1}{P_2} = 3$ faulty positions from $S_1 = [\Replica_{1,1}, \dots, \Replica_{1,5}]$ with $\Faulty{S_1} = \{ \Replica_{1,1}, \Replica_{1,2}, \Replica_{1,3} \}$ and $S_2 = [\Replica_{2,1}, \dots, \Replica_{2,5} ]$ with $\Faulty{S_2} = \{\Replica_{2,1}\}$. To create $(P_1, P_2)$, we need to create $b_1 = 2$ pairs that are $1$-faulty. We have $\binom{m_1}{b_1} = \binom{3}{2} = 3$ sets of two faulty replicas in $S_1$ that we can choose, namely the sets $\{\Replica_{1,1}, \Replica_{1,2} \}$, $\{\Replica_{1,1}, \Replica_{1,3} \}$, and $\{\Replica_{1,2}, \Replica_{1,3} \}$. Likewise, we have $\binom{n-m_2}{b_1} = \binom{4}{2} = 6$  sets of two non-faulty replicas in $S_2$ that we can choose. Assume we choose $T_1 = \{ \Replica_{1,1}, \Replica_{1,3} \}$ from $S_1$ and $T_2 = \{\Replica_{2,4}, \Replica_{2,5} \}$ from $S_2$. The two replicas in $T_1$ can be ordered in $\n{T_1}! = 2! = 2$ ways, namely $[\Replica_{1,1}, \Replica_{1,3}]$ and $[\Replica_{1,3}, \Replica_{1,1}]$. Likewise, the two replicas in $T_2$ can be ordered in $\n{T_2}! = 2! = 2$ ways. Hence, we can construct $2 \cdot 2 = 4$ distinct lists of type A out of this single choice for $T_1$ and $T_2$, and the sequences $S_1$ and $S_2$ provide us with $\binom{m_1}{b_1} \binom{n-m_2}{b_1} = 18$ distinct choices for $T_1$ and $T_2$. We conclude that we can construct $72$ distinct lists of type A from $S_1$ and $S_2$.
\end{example}

By construction, lists of type~\ref{lt:b1} and type~\ref{lt:b2} cannot utilize the same replicas from $S_1$ or $S_2$. After choosing $b_1 + b_2$ replicas in $S_1$ and $S_2$ for the construction of lists of type~\ref{lt:b1} and~\ref{lt:b2}, the remaining $b_{1,2}$ faulty replicas in $S_1$ and $S_2$ are all used for constructing lists of type~\ref{lt:b12}. As we can order these remaining values from $S_1$ and $S_2$ in $b_{1,2}!$ distinct ways, we can construct $b_{1,2}!^2$ distinct lists of type~\ref{lt:b12} (per choice of lists of type~\ref{lt:b1} and~\ref{lt:b2}). Likewise, the remaining $n- k$ non-faulty replicas in $S_1$ and $S_2$ are all used for constructing lists of type~\ref{lt:nk}, and we can construct $(n-k)!^2$ distinct lists of type~\ref{lt:nk} (per choice of lists of type~\ref{lt:b1} and~\ref{lt:b2}).

As the final steps, we merge lists of type~\ref{lt:b1} and~\ref{lt:b2} into lists of type AB. We can do so in $\LM{b_1}{b_2}$ ways and the resultant lists have size $b_1 + b_2$. Next, we merge lists of type AB and~\ref{lt:b12} into lists of type ABC. We can do so in $\LM{b_1 + b_2}{b_{1,2}}$ ways and the resultant lists have size $k$. Finally, we merge list of type ABC and~\ref{lt:nk} together, which we can do in $\LM{k}{n-k}$ ways. From this construction, we derive that $\FC{n}{m_1}{m_2}{k}$ is equivalent to
\begin{multline*}
    \FC{n}{m_1}{m_2}{k} = b_1!^2 \binom{m_1}{b_1}\binom{n-m_2}{b_1} b_2!^2 \binom{n-m_1}{b_2}\binom{m_2}{b_2} \cdot{}\\ \LM{b_1}{b_2} b_{1,2}!^2 \LM{b_1 + b_2}{b_{1,2}} (n-k)!^2   \LM{k}{n-k},
\end{multline*}
which can be simplified to the following (see Appendix~\ref{app:closed_form} for details):

\begin{lemma}\label{lem:bad_combi}
Let $\max(m_1, m_2) \leq k \leq \min(n, m_1 + m_2)$ and let $b_1 = k - m_2$, $b_2 = k - m_1$, and $b_{1,2} = (m_1 + m_2) - k$. We have \[\FC{n}{m_1}{m_2}{k} = \frac{m_1!m_2!(n-m_1)!(n-m_2)n!}{b_1!b_2!b_{1,2}!(n-k)!}.\]
\end{lemma}

We combine Lemma~\ref{lem:gpt_prob_open} and Lemma~\ref{lem:bad_combi} to conclude

\begin{proposition}\label{prop:pt_upper}
Let $S_1$ and $S_2$ be lists of $n = \n{S_1} = \n{S_2}$ replicas with $m_1 = \f{S_1}$, $m_2 = \f{S_2}$, $b_1 = k - m_2$, $b_2 = k - m_1$, and $b_{1,2} = (m_1 + m_2) - k$. If $m_1 + m_2 < n$, then the non-faulty position trials problem has solution
\[ \PT{n}{m_1}{m_2} = \frac{1}{n!^2} \left(\sum_{k = \max(m_1, m_2)}^{m_1 + m_2}\ \frac{n}{n-k} \frac{m_1!m_2!(n-m_1)!(n-m_2)!n!}{b_1!b_2!b_{1,2}!(n-k)!}\right).\]
\end{proposition}

Finally, we use the non-faulty position trials problem to derive

\begin{proposition}\label{prop:lpcs_complex}
Let $\Cluster_1, \Cluster_2$ be disjoint clusters and let $\SF$ be a list-pair function with $(S_1, S_2) \GETS \SF(\Cluster_1, \Cluster_2)$ and $n = \n{S_1} = \n{S_2}$. If communication is synchronous and $\f{S_1} + \f{S_2} < n$, then the expected number of cluster-sending steps performed by \PROTOCALL{Plcs}{$\Cluster_1$, $\Cluster_2$, $v$, $\SF$} is less than $\PT{n}{\f{S_1}}{\f{S_2}}$.
\end{proposition}
\begin{proof}
Let $(P_1, P_2) \in \Permute{S_1} \times \Permute{S_2}$. We notice that \Name{Plcs} inspects positions in $P_1$ and $P_2$ in a different way than the non-faulty trials problem: at \lfref{fig:plcs}{choose}, positions are inspected one-by-one in a predetermined order and not fully at random (with replacement). Next, we will argue that $\PT{n}{\f{S_1}}{\f{S_2}}$ provides an upper bound on the expected number of cluster-sending steps regardless of these differences. Without loss of generality, we assume that $S_1$ and $S_2$ each have $n$ distinct replicas. Consequently, the pair $(P_1, P_2)$ represents a set $R$ of $n$ distinct replica pairs taken from $\Cluster_1 \times \Cluster_2$. We notice that each of the $n!$ permutations of $R$ is represented by a single pair $(P_1', P_2') \in \Permute{S_1} \times \Permute{S_2}$. 

Now consider the selection of positions in $(P_1, P_2)$ fully at random, but without replacement. This process will yield a list $[j_0, \dots, j_{n-1}] \in \Permute{[0, \dots, n-1]}$ of positions fully at random. Let $Q_i = [P_i[j_0], \dots, P_i[j_{n-1}]]$, $i \in \{1,2\}$. We notice that the pair $(Q_1, Q_2)$ also represents $R$ and we have $(Q_1, Q_2) \in \Permute{S_1} \times \Permute{S_2}$. Hence, by choosing a pair $(P_1, P_2) \in \Permute{S_1} \times \Permute{S_2}$, we choose set $R$ fully at random and, at the same time, we choose the order in which replica pairs in $R$ are inspected fully at random.

As the final step in showing that $\PT{n}{\f{S_1}}{\f{S_2}}$ is an upper-bound to the expected number of cluster-sending steps performed by \Name{Plcs}, we notice that the number of expected positions inspected in the non-faulty position trials problem decreases if we choose positions without replacement, as done by \Name{Plcs}.
 \end{proof}

\subsection{Practical Instances of \Name{Plcs}}\label{ss:complex_lcs}

As the last step in providing practical instances of \Name{Plcs}, we need to provide practical list-pair functions to be used in conjunction with \Name{Plcs}. We provide two such functions that address most practical environments. Let $\Cluster_1, \Cluster_2$ be disjoint clusters, let $n_{\min} = \min(\n{\Cluster_1}, \n{\Cluster_2})$, and let $n_{\max} = \max(\n{\Cluster_1}, \n{\Cluster_2})$. We provide list-pair functions 
    \begin{align*}
        \SFi(\Cluster_1, \Cluster_2) &\mapsto (\Repeat{n_{\min}}{\List{\Cluster_1}}, \Repeat{n_{\min}}{\List{\Cluster_2}}),\\
        \SFa(\Cluster_1, \Cluster_2) &\mapsto (\Repeat{n_{\max}}{\List{\Cluster_2}}, \Repeat{n_{\max}}{\List{\Cluster_2}}),
    \end{align*}
in which $\Repeat{n}{L}$ denotes the first $n$ values in the list obtained by repeating list $L$. Next, we illustrate usage of these functions:

\begin{example}
Consider clusters $\Cluster_1, \Cluster_2$ with $S_1 = \List{\Cluster_1} = [ \Replica_{1,1}, \dots, \Replica_{1,9}]$ and $S_2 = \List{\Cluster_2} = [ \Replica_{2,1}, \dots, \Replica_{2, 4}]$. We have $\SFi(\Cluster_1, \Cluster_2) = ([\Replica_{1,1}, \dots, \Replica_{1, 4}], \List{\Cluster_2})$ and $\SFa(\Cluster_1, \Cluster_2) = (\List{\Cluster_1}, [\Replica_{2,1}, \dots, \Replica_{2, 4},\Replica_{2,1}, \dots, \Replica_{2, 4},\Replica_{2,1}])$.
\end{example}

Next, we combine $\SFi$ and $\SFa$ with \Name{Plcs}, show that in practical environments $\SFi$ and $\SFa$ satisfy the requirements put on list-pair functions in Proposition~\ref{prop:lpcs} to guarantee termination and cluster-sending, and use these results to determine the expected constant complexity of the resulting instances of \Name{Plcs}.

\begin{theorem}\label{thm:plcs_main}
Let $\Cluster_1, \Cluster_2$ be disjoint clusters with synchronous communication.
\begin{enumerate}
\item If $\min(\n{\Cluster_1}, \n{\Cluster_2}) > 2\max(\f{\Cluster_1}, \f{\Cluster_2})$, $(S_1, S_2) \GETS \SFi(\Cluster_1, \Cluster_2)$, and $n = \n{S_1} = \n{S_2}$,  then $n > 2\f{S_1}$, $n > 2\f{S_2}$, $n > \f{S_1} + \f{S_2}$, and the expected number of cluster-sending steps performed by \PROTOCALL{Plcs}{$\Cluster_1$, $\Cluster_2$, $v$, $\SFi$} is upper bounded by $4$.
\item If $\min(\n{\Cluster_1}, \n{\Cluster_2}) > 3\max(\f{\Cluster_1}, \f{\Cluster_2})$, $(S_1, S_2) \GETS \SFi(\Cluster_1, \Cluster_2)$, and $n = \n{S_1} = \n{S_2}$,  then $n > 3\f{S_1}$, $n > 3\f{S_2}$, $n > \f{S_1} + \f{S_2}$, and the expected number of cluster-sending steps performed by \PROTOCALL{Plcs}{$\Cluster_1$, $\Cluster_2$, $v$, $\SFi$} is upper bounded by $2\frac{1}{4}$.
\item If $\n{\Cluster_1} > 3\f{\Cluster_1}$, $\n{\Cluster_2} > 3\f{\Cluster_2}$, $(S_1, S_2) \GETS \SFa(\Cluster_1, \Cluster_2)$, and $n = \n{S_1} = \n{S_2}$, then either $\n{\Cluster_1} \geq \n{\Cluster_2}$, $n > 3\f{S_1}$, and $n > 2\f{S_2}$; or $\n{\Cluster_2} \geq \n{\Cluster_1}$, $n > 2\f{S_1}$, and $n > 3\f{S_2}$. In both cases, $n > \f{S_1} + \f{S_2}$ and the expected number of cluster-sending steps performed by \PROTOCALL{Plcs}{$\Cluster_1$, $\Cluster_2$, $v$, $\SFa$} is upper bounded by $3$.
\end{enumerate}
Each of these instance of \Name{Plcs} results in cluster-sending $v$ from $\Cluster_1$ to $\Cluster_2$.
\end{theorem}
\begin{proof}
First, we prove the properties of $\SFi$ and $\SFa$ claimed in the three statements of the theorem. In the first and second statement of the theorem, we have $\min(\n{\Cluster_1}, \n{\Cluster_2}) > c \max(\f{\Cluster_1}, \f{\Cluster_2})$, $c \in \{2,3\}$. Let $(S_1, S_2) \GETS \SFi(\Cluster_1, \Cluster_2)$ and $n = \n{S_1} = \n{S_2}$. By definition of $\SFi$, we have $n = \min(\n{\Cluster_1}, \n{\Cluster_2})$, in which case $S_i$, $i \in \{1,2\}$, holds $n$ distinct replicas from $\Cluster_i$. Hence, we have $\f{\Cluster_i} \geq \f{S_i}$ and, as $n > c \max(\f{\Cluster_1}, \f{\Cluster_2}) \geq c\f{\Cluster_i}$, also $n > c\f{S_i}$. Finally, as $n > 2\f{S_1}$ and $n > 2\f{S_2}$, also $2n > 2\f{S_1} + 2\f{S_2}$ and $n > \f{S_1} + \f{S_2}$ holds.

In the last statement of the theorem, we have $\n{\Cluster_1} > 3\f{\Cluster_1}$ and $\n{\Cluster_2} > 3\f{\Cluster_2}$. Without loss of generality, we  assume $\n{\Cluster_1} \geq \n{\Cluster_2}$. Let $(S_1, S_2) \GETS \SFa(\Cluster_1, \Cluster_2)$ and $n = \n{S_1} = \n{S_2}$. By definition of $\SFa$, we have $n = \max(\n{\Cluster_1}, \n{\Cluster_2}) = \n{\Cluster_1}$. As $n = \n{\Cluster_1}$, we have  $S_1 = \List{\Cluster_1}$. Consequently, we also have $\f{S_1} = \f{\Cluster_1}$ and, hence, $\n{S_1} > 3\f{\Cluster_1}$. Next, we will show that $\n{S_2} > 2\f{S_2}$. Let $q = \n{\Cluster_1} \div \n{\Cluster_2}$ and $r = \n{\Cluster_1} \bmod \n{\Cluster_2}$. We note that $\Repeat{n}{\List{\Cluster_2}}$ contains $q$ full copies of $\List{\Cluster_2}$ and one partial copy of $\List{\Cluster_2}$. Let $T \subset \Cluster_2$ be the set of replicas in this partial copy. By construction, we have $\n{S_2} = q\n{\Cluster_2} + r > q3\f{\Cluster_2} + \f{T} + \nf{T}$ and $\f{S_2} = q\f{\Cluster_2} + \f{T}$ with $\f{T} \leq \min(\f{\Cluster_2}, r)$. As $q > 1$ and $\f{\Cluster_2} \geq \f{T}$, we have $q\f{\Cluster_2} \geq \f{\Cluster_2} \geq \f{T}$. Hence, $\n{S_2} >  3q\f{\Cluster_2} + \f{T} + \nf{T} > 2q\f{\Cluster_2} + \f{\Cluster_2} + \f{T} + \nf{T} \geq  2(q\f{\Cluster_2} + \f{T}) + \nf{T} \geq 2\f{S_2}$. Finally, as $n > 3\f{S_1}$ and $n > 2\f{S_2}$, also $2n > 3\f{S_1} + 2\f{S_2}$ and $n > \f{S_1} + \f{S_2}$ holds.

Now, we prove the upper bounds on the expected number of cluster-sending steps for \PROTOCALL{Plcs}{$\Cluster_1$, $\Cluster_2$, $v$, $\SFi$} with $\min(\n{\Cluster_1}, \n{\Cluster_2}) > 2\max(\f{\Cluster_1}, \f{\Cluster_2})$.  By Proposition~\ref{prop:lpcs_complex}, the expected number of cluster-sending steps is upper bounded by $\PT{n}{\f{S_1}}{\f{S_2}}$. In the worst case, we have $n = 2f + 1$ with $f = \f{S_1} = \f{S_2}$. Hence, the expected number of cluster-sending steps is upper bounded by $\PT{2f+1}{f}{f}$, $f \geq 0$. We claim that $\PT{2f+1}{f}{f}$ simplifies to $\PT{2f+1}{f}{f} = 4 - \dsfrac{2}{(f+1)} - \dsfrac{f!^2}{(2f)!}$ (see Appendix~\ref{app:app_simply} for details). Hence, for all $S_1$ and $S_2$, we have $\PT{n}{\f{S_1}}{\f{S_2}} < 4$. An analogous argument can be used to prove the other upper bounds.
\end{proof}

Note that the third case of Theorem~\ref{thm:plcs_main} corresponds with cluster-sending between arbitrary-sized resilient clusters that each operate using Byzantine fault-tolerant consensus protocols.

\begin{remark}
The upper bounds on the expected-case complexity of instances of \Name{Plcs} presented in Theorem~\ref{thm:plcs_main} match the upper bounds for  \Name{Pcs} presented in Corollary~\ref{cor:cs}. This does not imply that the expected-case complexity for these protocols is the same, however, as the probability distributions that yield these expected-case complexities are very different. To see this, consider a system in which all clusters have $n$ replicas of which $f$, $n = 2f+1$, are faulty. Next, we denote the expected number of cluster-sending steps of protocol $P$ by $\mathbf{E}_{P}$, and we have
\begin{align*}
    \mathbf{E}_{\Name{Pcs}} &= \frac{(2f+1)^2}{(f+1)^2} = 4 - \frac{4f+3}{(f+1)^2};\\
    \mathbf{E}_{\Name{Plcs}} &= \PT{2f+1}{f}{f} = 4 - \frac{2}{(f+1)} - \frac{f!^2}{(2f)!}.
\end{align*}
In Figure~\ref{fig:expected_plot} (Section~\ref{sec:related}), \emph{left} and \emph{middle}, we have illustrated this difference by plotting the expected-case complexity of \Name{Pcs} and \Name{Plcs} for systems with equal-sized clusters. In practice, we see that the expected-case complexity for \Name{Pcs} is slightly lower than the expected-case complexity for \Name{Plcs}.
\end{remark}

\section{Dealing with Unreliable and Asynchronous Communication}\label{sec:async}
In the previous sections, we introduced \Name{Pcs}, \Name{Ppcs}, and \Name{Plcs}, three probabilistic cluster-sending protocols with expected constant message complexity. As presented, these protocols are designed to operate in a synchronous environment: if a cluster $\Cluster_1$ wants to send a value $v$ to $\Cluster_2$, then the replicas in $\NonFaulty{\Cluster_1}$ use time-based decisions to determine whether a cluster-sending step was successful. Next, we consider their usage in environments with asynchronous inter-cluster communication due to which messages can get arbitrary delayed, duplicated, or dropped. 

We notice that the presented protocols \emph{only} depend on synchronous communication to minimize communication: at the core of the correctness of \Name{Pcs}, \Name{Ppcs}, and \Name{Plcs} is the cluster-sending step performed by \Name{cs-step}, which does not make any assumptions on communication (Proposition~\ref{prop:cs_step}). Consequently, \Name{Pcs}, \Name{Ppcs}, and \Name{Plcs} can easily be generalized to operate in environments with asynchronous communication:
\begin{enumerate}
\item First, we observe that message duplication and out-of-order delivery has no impact on the cluster-sending step performed by \Name{cs-step}. Hence, we do not need to take precautions against such asynchronous behavior.
\item If communication is asynchronous, but reliable (messages do not get lost, but can get duplicated, be delivered out-of-order, or get arbitrarily delayed), both \Name{Ppcs} and \Name{Plcs} will be able to always perform cluster-sending in a finite number of steps. If communication becomes unreliable, however, messages sent between non-faulty replicas can get lost and all cluster-sending steps can fail. To deal with this, replicas in $\Cluster_1$ simply continue cluster-sending steps until a step succeeds, which will eventually happen in an expected constant number steps whenever communication becomes reliable again.
\item If communication is asynchronous, then messages can get arbitrarily delayed. Fortunately, practical environments operate with large periods of reliable communication in which the majority of the messages arrive within some bounded delay unknown to $\Cluster_1$ and $\Cluster_2$. Hence, replicas in $\Cluster_1$ can simply assume some delay $\delta$. If this delay is too short, then a cluster-sending step can \emph{appear to fail} simply because the proof of receipt is still under way. In this case, cluster-sending will still be achieved when the proof of receipt arrives, but spurious cluster-sending steps can be initiated in the meantime. To reduce the number of such spurious cluster-sending steps, all non-faulty replicas in $\Cluster_1$ can use \emph{exponential backup} for the delay such that the $i$-th cluster-sending step must finish $\delta 2^i$ time units after $\Cluster_1$ reached agreement on sending $v$ to $\Cluster_2$.
\item Finally, asynchronous environments often necessitate rather high assumptions on the message delay $\delta$. Consequently, the duration of a single failed cluster-sending step performed by \Name{cs-step} will be high. Here, a trade-off can be made between \emph{message complexity} and \emph{duration} by starting several rounds of the cluster-sending step at once. E.g., when communication is sufficiently reliable, then all three protocols are expected to finish in four rounds or less, due to which starting four rounds initially will sharply reduce the duration of the protocol with only a constant increase in expected message complexity.
\end{enumerate}

\section{Comparison with Related Work}\label{sec:related}

In the previous sections, we presented our cluster-sending protocols \Name{Pcs}, \Name{Ppcs}, and \Name{Pcs}. Next, we compare them with existing approaches towards communication between resilient clusters.  Although there is abundant literature on distributed systems and on consensus-based resilient systems (e.g.,~\cite{scaling,wild,untangle,book,encybd,distdb,distalgo,distbook}), there is only limited work on communication \emph{between} resilient systems~\cite{chainspace,vldb,disc_csp} (see also Table~\ref{tbl:summary}).

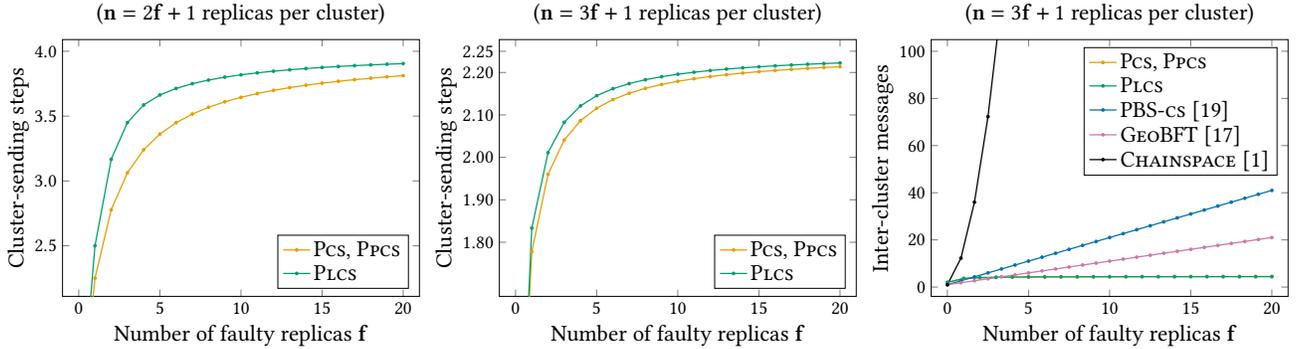
\begin{figure}
    \centering
    \makebox[0pt]{
    \begin{tikzpicture}[plot]
        \begin{axis}[title={($\mathbf{n} = 2\mathbf{f} + 1$ replicas per cluster)},
                     xlabel={Number of faulty replicas $\mathbf{f}$},
                     ylabel={\smash{Cluster-sending steps}},
                     ymin=2.2,ymax=4,xmin=0,xmax=20,
                     legend pos=south east,
                     ytick={2.5,3,3.5,4},
                     y tick label style={
                        /pgf/number format/precision=1,
                        /pgf/number format/fixed,
                        /pgf/number format/fixed zerofill
                    }]
            \addplot table[x={f},y={pcs}] {\dataCaseTwo};
            \addplot table[x={f},y={plcs}] {\dataCaseTwo};
            \legend{\Name{Pcs}{, }\Name{Ppcs} ,\Name{Plcs}};
        \end{axis}
    \end{tikzpicture}\quad
    \begin{tikzpicture}[plot]
        \begin{axis}[title={($\mathbf{n} = 3\mathbf{f} + 1$ replicas per cluster)},
                     xlabel={Number of faulty replicas $\mathbf{f}$},
                     ylabel={\smash{Cluster-sending steps}},
                     ymin=1.7,ymax=2.25,xmin=0,xmax=20,
                     legend pos=south east,
                     ytick={1.8,1.9,2,2.1,2.2,2.25},
                     y tick label style={
                        /pgf/number format/precision=2,
                        /pgf/number format/fixed,
                        /pgf/number format/fixed zerofill
                    }]
            \addplot table[x={f},y={pcs}] {\dataCaseThree};
            \addplot table[x={f},y={plcs}] {\dataCaseThree};
            \legend{\Name{Pcs}{, }\Name{Ppcs},\Name{Plcs}};
        \end{axis}
    \end{tikzpicture}\quad
    \begin{tikzpicture}[plot]
        \begin{axis}[title={($\mathbf{n} = 3\mathbf{f} + 1$ replicas per cluster)},
                     xlabel={Number of faulty replicas $\mathbf{f}$},
                     ylabel={\smash{Inter-cluster messages}},
                     xmin=0,xmax=20,domain=0:20,ymax=100,
                     legend pos=north east]
            \addplot table[x={f},y expr={2 * \thisrow{pcs}}] {\dataCaseThree};
            \addplot table[x={f},y expr={2 * \thisrow{plcs}}] {\dataCaseThree};
            \addplot {x+x+1};
            \addplot {x+1};
            \addplot {(3*x+1)*(3*x+1)};
            \legend{\Name{Pcs}{, }\Name{Ppcs} ,\Name{Plcs},\Name{PBS-cs}~\cite{disc_csp},\Name{GeoBFT}~\cite{vldb},\Name{Chainspace}~\cite{chainspace}};
        \end{axis}
    \end{tikzpicture}}
    \caption{Comparison of the expected-case complexity of \Name{Plcs} and \Name{Pcs} (\emph{left} and \emph{middle}) and a comparison with the complexity of cluster-sending protocols proposed in related work (\emph{right}).}\label{fig:expected_plot}
\end{figure}

Consider sending a value $v$ between equal-sized clusters $\Cluster_1$ and $\Cluster_2$. First, the \emph{multicast-based cluster-sending protocol} of \Name{Chainspace}~\cite{chainspace} requires reliable communication and can perform cluster-sending using $\n{\Cluster_1}\n{\Cluster_2}$ messages. Next, the \emph{worst-case optimal cluster-sending protocols} of Hellings et al.~\cite{disc_csp} also require reliable communication, but can cluster-send using only  $\f{\Cluster_1} + \f{\Cluster_2} + 1$ messages. Finally, the \emph{global sharing protocol} of \Name{GeoBFT}~\cite{vldb} assumes that each cluster uses a primary-backup consensus protocol (e.g., \Name{Pbft}~\cite{pbftj}) and optimizes for the case in which the coordinating primary of $\Cluster_1$ is non-faulty. In this optimistic case, \Name{GeoBFT} can perform cluster-sending using only $\f{\Cluster_2} + 1$ messages. To deal with faulty primaries and unreliable communication, \Name{GeoBFT} employs a costly remote view-change protocol, however. Finally, we notice that \emph{cluster-sending} can be solved using well-known Byzantine primitives such as consensus, interactive consistency, and Byzantine broadcasts~\cite{pbftj,netbound,dolevstrong,byzgen,generals}. These primitives are much more costly than the above cluster-sending protocols, however, and typically require huge amounts of costly communication between all involved replicas. In Figure~\ref{fig:expected_plot}, \emph{right}, we compare the inter-cluster message complexity of existing cluster-sending protocols to the protocols presented in this work. Clearly, our protocols sharply reduce inter-cluster communication, as they have an expected constant complexity. Furthermore, our protocols can effectively deal with unreliable asynchronous communication with low cost.

In parallel to the development of traditional resilient systems and permissioned blockchains, there has been promising work on sharding in permissionless blockchains such as \Name{Bitcoin}~\cite{bitcoin} and \Name{Ethereum}~\cite{ethereum}. Examples include techniques for enabling reliable cross-chain coordination via sidechains, blockchain relays, atomic swaps, atomic commitment, and cross-chain deals~\cite{blockchaindb,btcrelay,atomswap,cross_deal,cosmos,polkadot,atomic_comm}. Unfortunately, these techniques are deeply intertwined with the design goals of permissionless blockchains in mind (e.g., cryptocurrency-oriented), and are not readily applicable to traditional consensus-based Byzantine clusters.

\section{Conclusion}\label{sec:concl}

In this paper, we presented several probabilistic cluster-sending protocols that can facilitate communication between Byzantine fault-tolerant clusters with expected constant communication between clusters. For practical environments, our protocols can support worst-case linear communication between clusters, which is optimal, and deal with asynchronous and unreliable communication. The low cost of our cluster-sending protocols enables the development and deployment of high-performance systems that are constructed out of Byzantine fault-tolerant clusters, e.g., fault-resilient geo-aware sharded data processing systems.

\bibliography{resources}

\begin{thebibliography}{10}

\bibitem{chainspace}
Mustafa Al-Bassam, Alberto Sonnino, Shehar Bano, Dave Hrycyszyn, and George
  Danezis.
\newblock Chainspace: A sharded smart contracts platform, 2017.
\newblock URL: \texttt{http://arxiv.org/abs/1708.03778}.

\bibitem{scaling}
Christian Berger and Hans~P. Reiser.
\newblock Scaling byzantine consensus: A broad analysis.
\newblock In {\em Proceedings of the 2nd Workshop on Scalable and Resilient
  Infrastructures for Distributed Ledgers}, pages 13--18. ACM, 2018.
\newblock \href {https://doi.org/10.1145/3284764.3284767}
  {\path{doi:10.1145/3284764.3284767}}.

\bibitem{coin}
Gabi Bracha and Ophir Rachman.
\newblock Randomized consensus in expected {$\BigO{(n^2\log n)}$} operations.
\newblock In {\em Distributed Algorithms}, pages 143--150. Springer Berlin
  Heidelberg, 1992.
\newblock \href {https://doi.org/10.1007/BFb0022443}
  {\path{doi:10.1007/BFb0022443}}.

\bibitem{bdrandcoin}
Christian Cachin, Klaus Kursawe, Frank Petzold, and Victor Shoup.
\newblock Secure and efficient asynchronous broadcast protocols.
\newblock In {\em Advances in Cryptology --- CRYPTO 2001}, pages 524--541.
  Springer Berlin Heidelberg, 2001.
\newblock \href {https://doi.org/10.1007/3-540-44647-8_31}
  {\path{doi:10.1007/3-540-44647-8_31}}.

\bibitem{wild}
Christian Cachin and Marko Vukolic.
\newblock Blockchain consensus protocols in the wild (keynote talk).
\newblock In {\em 31st International Symposium on Distributed Computing},
  volume~91 of {\em Leibniz International Proceedings in Informatics (LIPIcs)},
  pages 1:1--1:16. Schloss Dagstuhl--Leibniz-Zentrum fuer Informatik, 2017.
\newblock \href {https://doi.org/10.4230/LIPIcs.DISC.2017.1}
  {\path{doi:10.4230/LIPIcs.DISC.2017.1}}.

\bibitem{pbftj}
Miguel Castro and Barbara Liskov.
\newblock Practical byzantine fault tolerance and proactive recovery.
\newblock {\em ACM Transactions on Computer Systems}, 20(4):398--461, 2002.
\newblock \href {https://doi.org/10.1145/571637.571640}
  {\path{doi:10.1145/571637.571640}}.

\bibitem{ahl}
Hung Dang, Tien Tuan~Anh Dinh, Dumitrel Loghin, Ee-Chien Chang, Qian Lin, and
  Beng~Chin Ooi.
\newblock Towards scaling blockchain systems via sharding.
\newblock In {\em Proceedings of the 2019 International Conference on
  Management of Data}, pages 123--140. ACM, 2019.
\newblock \href {https://doi.org/10.1145/3299869.3319889}
  {\path{doi:10.1145/3299869.3319889}}.

\bibitem{untangle}
Tien Tuan~Anh Dinh, Rui Liu, Meihui Zhang, Gang Chen, Beng~Chin Ooi, and
  Ji~Wang.
\newblock Untangling blockchain: A data processing view of blockchain systems.
\newblock {\em IEEE Trans. Knowl. Data Eng.}, 30(7):1366--1385, 2018.
\newblock \href {https://doi.org/10.1109/TKDE.2017.2781227}
  {\path{doi:10.1109/TKDE.2017.2781227}}.

\bibitem{netbound}
D.~Dolev.
\newblock Unanimity in an unknown and unreliable environment.
\newblock In {\em 22nd Annual Symposium on Foundations of Computer Science},
  pages 159--168. IEEE, 1981.
\newblock \href {https://doi.org/10.1109/SFCS.1981.53}
  {\path{doi:10.1109/SFCS.1981.53}}.

\bibitem{dolevstrong}
D.~Dolev and H.~Strong.
\newblock Authenticated algorithms for byzantine agreement.
\newblock {\em SIAM Journal on Computing}, 12(4):656--666, 1983.
\newblock \href {https://doi.org/10.1137/0212045} {\path{doi:10.1137/0212045}}.

\bibitem{byzgen}
Danny Dolev.
\newblock The byzantine generals strike again.
\newblock {\em Journal of Algorithms}, 3(1):14--30, 1982.
\newblock \href {https://doi.org/10.1016/0196-6774(82)90004-9}
  {\path{doi:10.1016/0196-6774(82)90004-9}}.

\bibitem{min_sync}
Danny Dolev, Cynthia Dwork, and Larry Stockmeyer.
\newblock On the minimal synchronism needed for distributed consensus.
\newblock {\em Journal of the ACM}, 34(1):77--97, 1987.
\newblock \href {https://doi.org/10.1145/7531.7533}
  {\path{doi:10.1145/7531.7533}}.

\bibitem{partsync}
Cynthia Dwork, Nancy Lynch, and Larry Stockmeyer.
\newblock Consensus in the presence of partial synchrony.
\newblock {\em Journal of the ACM}, 35(2):288--323, 1988.
\newblock \href {https://doi.org/10.1145/42282.42283}
  {\path{doi:10.1145/42282.42283}}.

\bibitem{blockchaindb}
Muhammad El-Hindi, Carsten Binnig, Arvind Arasu, Donald Kossmann, and Ravi
  Ramamurthy.
\newblock {BlockchainDB}: A shared database on blockchains.
\newblock {\em Proc. VLDB Endow.}, 12(11):1597--1609, 2019.
\newblock \href {https://doi.org/10.14778/3342263.3342636}
  {\path{doi:10.14778/3342263.3342636}}.

\bibitem{btcrelay}
Ethereum Foundation.
\newblock {BTC Relay}: A bridge between the bitcoin blockchain \& ethereum
  smart contracts, 2017.
\newblock URL: \texttt{http://btcrelay.org}.

\bibitem{book}
Suyash Gupta, Jelle Hellings, and Mohammad Sadoghi.
\newblock {\em Fault-Tolerant Distributed Transactions on Blockchain}.
\newblock Synthesis Lectures on Data Management. Morgan \& Claypool, 2021.
\newblock \href {https://doi.org/10.2200/S01068ED1V01Y202012DTM065}
  {\path{doi:10.2200/S01068ED1V01Y202012DTM065}}.

\bibitem{vldb}
Suyash Gupta, Sajjad Rahnama, Jelle Hellings, and Mohammad Sadoghi.
\newblock {ResilientDB}: Global scale resilient blockchain fabric.
\newblock {\em Proceedings of the VLDB Endowment}, 13(6):868--883, 2020.
\newblock \href {https://doi.org/10.14778/3380750.3380757}
  {\path{doi:10.14778/3380750.3380757}}.

\bibitem{encybd}
Suyash Gupta and Mohammad Sadoghi.
\newblock {\em Blockchain Transaction Processing}, pages 1--11.
\newblock Springer International Publishing, 2018.
\newblock \href {https://doi.org/10.1007/978-3-319-63962-8_333-1}
  {\path{doi:10.1007/978-3-319-63962-8_333-1}}.

\bibitem{disc_csp}
Jelle Hellings and Mohammad Sadoghi.
\newblock Brief announcement: The fault-tolerant cluster-sending problem.
\newblock In {\em 33rd International Symposium on Distributed Computing (DISC
  2019)}, volume 146 of {\em Leibniz International Proceedings in Informatics
  (LIPIcs)}, pages 45:1--45:3. Schloss Dagstuhl--Leibniz-Zentrum fuer
  Informatik, 2019.
\newblock \href {https://doi.org/10.4230/LIPIcs.DISC.2019.45}
  {\path{doi:10.4230/LIPIcs.DISC.2019.45}}.

\bibitem{byshard}
Jelle Hellings and Mohammad Sadoghi.
\newblock Byshard: Sharding in a byzantine environment.
\newblock {\em Proceedings of the VLDB Endowment}, 14(11):2230--2243, 2021.
\newblock \href {https://doi.org/10.14778/3476249.3476275}
  {\path{doi:10.14778/3476249.3476275}}.

\bibitem{atomswap}
Maurice Herlihy.
\newblock Atomic cross-chain swaps.
\newblock In {\em Proceedings of the 2018 ACM Symposium on Principles of
  Distributed Computing}, pages 245--254. ACM, 2018.
\newblock \href {https://doi.org/10.1145/3212734.3212736}
  {\path{doi:10.1145/3212734.3212736}}.

\bibitem{blockchain_dist}
Maurice Herlihy.
\newblock Blockchains from a distributed computing perspective.
\newblock {\em Communications of the ACM}, 62(2):78--85, 2019.
\newblock \href {https://doi.org/10.1145/3209623} {\path{doi:10.1145/3209623}}.

\bibitem{cross_deal}
Maurice Herlihy, Barbara Liskov, and Liuba Shrira.
\newblock Cross-chain deals and adversarial commerce.
\newblock {\em Proc. VLDB Endow.}, 13(2):100--113, 2019.
\newblock \href {https://doi.org/10.14778/3364324.3364326}
  {\path{doi:10.14778/3364324.3364326}}.

\bibitem{cosmos}
Jae Kwon and Ethan Buchman.
\newblock Cosmos whitepaper: A network of distributed ledgers, 2019.
\newblock URL: \texttt{https://cosmos.network/cosmos-whitepaper.pdf}.

\bibitem{paxossimple}
Leslie Lamport.
\newblock Paxos made simple.
\newblock {\em ACM SIGACT News, Distributed Computing Column 5}, 32(4):51--58,
  2001.
\newblock \href {https://doi.org/10.1145/568425.568433}
  {\path{doi:10.1145/568425.568433}}.

\bibitem{generals}
Leslie Lamport, Robert Shostak, and Marshall Pease.
\newblock The byzantine generals problem.
\newblock {\em ACM Transactions on Programming Languages and Systems},
  4(3):382--401, 1982.
\newblock \href {https://doi.org/10.1145/357172.357176}
  {\path{doi:10.1145/357172.357176}}.

\bibitem{bitcoin}
Satoshi Nakamoto.
\newblock Bitcoin: A peer-to-peer electronic cash system.
\newblock URL: \texttt{https://bitcoin.org/ en/bitcoin-paper}.

\bibitem{bit_pedigree}
Arvind Narayanan and Jeremy Clark.
\newblock Bitcoin's academic pedigree.
\newblock {\em Communications of the ACM}, 60(12):36--45, 2017.
\newblock \href {https://doi.org/10.1145/3132259} {\path{doi:10.1145/3132259}}.

\bibitem{distdb}
M.~Tamer {\"O}zsu and Patrick Valduriez.
\newblock {\em Principles of Distributed Database Systems}.
\newblock Springer, 2020.
\newblock \href {https://doi.org/10.1007/978-3-030-26253-2}
  {\path{doi:10.1007/978-3-030-26253-2}}.

\bibitem{consbound}
M.~Pease, R.~Shostak, and L.~Lamport.
\newblock Reaching agreement in the presence of faults.
\newblock {\em Journal of the ACM}, 27(2):228--234, 1980.
\newblock \href {https://doi.org/10.1145/322186.322188}
  {\path{doi:10.1145/322186.322188}}.

\bibitem{hypereal}
Michael Pisa and Matt Juden.
\newblock Blockchain and economic development: Hype vs. reality.
\newblock Technical report, Center for Global Development, 2017.
\newblock URL:
  \texttt{https://www.cgdev.org/publication/ blockchain-and-economic-development-hype-vs-reality}.

\bibitem{ringbft}
Sajjad Rahnama, Suyash Gupta, Rohan Sogani, Dhruv Krishnan, and Mohammad
  Sadoghi.
\newblock Ringbft: Resilient consensus over sharded ring topology, 2021.
\newblock URL: \texttt{https://arxiv.org/abs/2107.13047}.

\bibitem{idc}
David Reinsel, John Gantz, and John Rydning.
\newblock Data age 2025: The digitization of the world, from edge to core.
\newblock Technical report, IDC, 2018.
\newblock URL:
  \texttt{https://www.seagate.com/files/www-content/ our-story/trends/files/idc-seagate-dataage-whitepaper.pdf}.

\bibitem{rsasign}
Victor Shoup.
\newblock Practical threshold signatures.
\newblock In {\em Advances in Cryptology --- EUROCRYPT 2000}, pages 207--220.
  Springer Berlin Heidelberg, 2000.
\newblock \href {https://doi.org/10.1007/3-540-45539-6_15}
  {\path{doi:10.1007/3-540-45539-6_15}}.

\bibitem{distalgo}
Gerard Tel.
\newblock {\em Introduction to Distributed Algorithms}.
\newblock Cambridge University Press, 2nd edition, 2001.

\bibitem{distbook}
Maarten van Steen and Andrew~S. Tanenbaum.
\newblock {\em Distributed Systems}.
\newblock Maarten van Steen, 3th edition, 2017.
\newblock URL: \texttt{https://www.distributed-systems.net/}.

\bibitem{ethereum}
Gavin Wood.
\newblock Ethereum: a secure decentralised generalised transaction ledger.
\newblock {EIP}-150 revision.
\newblock URL: \texttt{https://gavwood.com/paper.pdf}.

\bibitem{polkadot}
Gavin Wood.
\newblock {Polkadot}: vision for a heterogeneous multi-chain framework, 2016.
\newblock URL: \texttt{https:// polkadot.network/PolkaDotPaper.pdf}.

\bibitem{atomic_comm}
Victor Zakhary, Divyakant Agrawal, and Amr El~Abbadi.
\newblock Atomic commitment across blockchains.
\newblock {\em Proc. VLDB Endow.}, 13(9):1319--1331, 2020.
\newblock \href {https://doi.org/10.14778/3397230.3397231}
  {\path{doi:10.14778/3397230.3397231}}.

\end{thebibliography}

\appendix

\section{The proof of Proposition~\ref{prop:cs_step}}\label{app:first}
\begin{proof} We prove the three post-conditions separately.

            \textbf{(i)} 
                We assume that communication is reliable, $\Replica_1 \in \NonFaulty{\Cluster_1}$, and $\Replica_2 \in \NonFaulty{\Cluster_2}$. Hence, $\Replica_1$ sends message $m \GETS \SignMessage{send}{v,\ \Cluster_2}{\Cluster_1}$ to $\Replica_2$ (\lfref{fig:cs_step}{send_v}). In the receive phase (\lsfref{fig:cs_step}{receive_event}{send_confirm}), replica $\Replica_2$ \emph{receives} message $m$ from $\Replica_1$. Replica $\Replica_2$ uses local consensus on $m$ to replicate $m$ among all replicas $\Cluster_2$ and, along the way, to constructs a \emph{proof of receipt} $m_p \GETS \SignMessage{proof}{m}{\Cluster_ 2}$. As all replicas in $\NonFaulty{\Cluster_2}$ participate in this local consensus, all replicas in $\NonFaulty{\Cluster_2}$ will decide \Decide{receive} on $v$ from $\Cluster_1$. Finally, the proof $m_p$ is returned to $\Replica_1$. In the confirmation phase (\lsfref{fig:cs_step}{replied_event}{confirm}), replica $\Replica_1$ receives the proof of receipt $m_p$. Next, $\Replica_1$ uses local consensus on $m_p$ to replicate $m_p$ among all replicas in $\NonFaulty{\Cluster_1}$, after which all replicas in $\NonFaulty{\Cluster_1}$ decide \Decide{confirm} on sending $v$ to $\Cluster_2$

            \textbf{(ii)} A replica in $\NonFaulty{\Cluster_2}$ only decides \Decide{receive} on $v$ after consensus is reached on a message $m \GETS \SignMessage{send}{v,\ \Cluster_2}{\Cluster_1}$ (\lfref{fig:cs_step}{receive}). This message $m$ not only contains the value $v$, but also the identity of the recipient cluster $\Cluster_2$. Due to the usage of certificates and the pre-condition, the message $m$ cannot be created without the replicas in $\NonFaulty{\Cluster_1}$ deciding \Decide{agree} on sending $v$ to $\Cluster_2$.
            
            \textbf{(iii)}
                A replica in $\NonFaulty{\Cluster_1}$ only decides \Decide{confirm} on $v$ after consensus is reached on a \emph{proof of receipt} message $m_p \GETS \SignMessage{proof}{m}{\Cluster_ 2}$ (\lfref{fig:cs_step}{confirm}). This consensus step will complete for all replicas in $\Cluster_1$ whenever communication becomes reliable. Hence, all replicas in $\NonFaulty{\Cluster_1}$ will eventually decide \Decide{confirm} on $v$. Due to the usage of certificates, the message $m_p$ cannot be created without cooperation of the replicas in $\NonFaulty{\Cluster_2}$. The replicas in $\NonFaulty{\Cluster_2}$ only cooperate in constructing $m_p$ as part of the consensus step of \lfref{fig:cs_step}{receive_v}. Upon completion of this consensus step, all replicas in $\NonFaulty{\Cluster_2}$ will decide \Decide{receive} on $v$.
\end{proof}

\section{The proof of Lemma~\ref{lem:cs_prune}}\label{app:prune}
\begin{proof}
The statement of this Lemma assumes that the pre-conditions for any execution of \Name{cs-step}($\Replica_1$, $\Replica_2$, $v$) with $\Replica_1 \in \Cluster_1$ and $\Replica_2 \in \Cluster_2$ are established. Hence, by Proposition~\ref{prop:cs_step}, \Name{cs-step}($\Replica_1$, $\Replica_2$, $v$) will cluster-send $v$ if $\Replica_1 \in \NonFaulty{\Cluster_1}$ and $\Replica_2 \in \NonFaulty{\Cluster_2}$. If the cluster-sending step fails to cluster-send $v$, then one of the replicas involved must be faulty, proving the first property. Next, let $\Replica_1 \in \Cluster_1$ and consider a set $S \subseteq \Cluster_2$ of $\n{S} = \f{\Cluster_2} + 1$ replicas such that, for all $\Replica_2 \in S$, \Name{cs-step}($\Replica_1$, $\Replica_2$, $v$) fails to cluster-send $v$. Let $S' = S \difference \Faulty{\Cluster_2}$ be the non-faulty replicas in $S$. As $\n{S} > \f{\Cluster_2}$, we have $\n{S'} \geq 1$ and there exists a $\Replica_2' \in S'$. As $\Replica_2' \notin \Faulty{\Cluster_2}$ and \Name{cs-step}($\Replica_1$, $\Replica_2'$, $v$) fails to cluster-send $v$, we must have $\Replica_1 \in \Faulty{\Cluster_1}$ by the first property, proving the second property. An analogous argument proves the third property.
\end{proof}

\section{The proof of Lemma~\ref{lem:list_merge}}\label{app:list_merge}

To get the intuition behind the closed form of Lemma~\ref{lem:list_merge}, we take a quick look at the combinatorics of \emph{list-merging}. Notice that we can merge lists $S$ and $T$ together by either first taking an element from $S$ or first taking an element from $T$. This approach towards list-merging yields the following recursive solution to the list-merge problem:
\[ \LM{v}{w} = \begin{cases} \LM{v-1}{w} + \LM{v}{w-1} &\text{if $v > 0$ and $w > 0$};\\
                             1 &\text{if $v = 0$ or $w = 0$}.
               \end{cases} \]
Consider lists $S$ and $T$ with $\abs{S} = v$ and $\abs{T} = w$ distinct values. We have $\abs{\Permute{S}} = v!$, $\abs{\Permute{T}} = w!$, and $\abs{\Permute{S \union T}} = (v+w)!$. We observe that every list-merge of $(P_S, P_T) \in \Permute{S} \times \Permute{T}$ is a unique value in $\Permute{S \union T}$. Furthermore, every value in $\Permute{S \union T}$ can be constructed by such a list-merge. As we have $\abs{\Permute{S} \times \Permute{T}} = v!w!$, we derive the closed form 
\[\LM{v}{w} = \frac{(v + w)!}{(v!w!)}\]
of Lemma~\ref{lem:list_merge}. Next, we formally prove this closed form.

\begin{proof}
We prove this by induction. First, the base cases $\LM{0}{w}$ and $\LM{v}{0}$. We have
\begin{align*}
    \LM{0}{w} &= \frac{(0 + w)!}{0! w!} = \frac{w!}{w!} = 1; \\
    \LM{v}{0} &= \frac{(v + 0)!}{v! 0!} = \frac{v!}{v!} = 1.
\end{align*}
Next, we assume that the statement of the lemma holds for all non-negative integers $v', w'$ with $0 \leq v' + w' \leq j$. Now consider non-negative integers $v, w$ with $v + w = j + 1$. We assume that $v > 0$ and $w > 0$, as otherwise one of the base cases applies. Hence, we have
\begin{align*}
    \LM{v}{w} &= \LM{v-1}{w} + \LM{v}{w-1}.
\intertext{We apply the induction hypothesis on the terms $\LM{v-1}{w}$ and $\LM{v}{w-1}$ and obtain}
    \LM{v}{w} &= \left(\frac{((v-1)+ w)!}{(v - 1)! w!}\right) + \left(\frac{(v + (w - 1))!}{v! (w-1)!}\right).
\intertext{Next, we apply $x = x(x-1)!$ and simplify the result to obtain}
    \LM{v}{w} &= \left(\frac{v (v + w - 1)!}{v! w!}\right) + \left(\frac{w (v + w - 1)!}{v! w!}\right)\\
              &= \left(\frac{(v + w) (v + w - 1)!}{v! w!}\right) = \frac{(v+w)!}{v! w!},
\end{align*}
which completes the proof.
\end{proof}

\section{The proof of Lemma~\ref{lem:bad_combi}}

\begin{proof}
We write $f(n, m_1, m_2, k)$ for the closed form in the statement of this lemma and we prove the statement of this lemma by induction. First, the base case $\FC{0}{0}{0}{0}$. In this case, we have $n = m_1 = m_2 = k = 0$ and, hence, $b_1 = b_2 = b_{1,2} = 0$, and we conclude $f(0, 0, 0, 0) = 1 = \FC{0}{0}{0}{0}$.

Now assume $\FC{n'}{m_1'}{m_2'}{k'} = f(n', m_1', m_2', k')$ for all $n' < n$ and all $k'$ with $\max(m_1', m_2') \leq k' \leq \min(n', m_1' + m_2')$. Next, we prove $\FC{n}{m_1}{m_2}{k} = f(n, m_1, m_2, k)$ with $\max(m_1, m_2) \leq k \leq \min(n, m_1 + m_2)$. We use the shorthand $\mathbb{G} = \FC{n}{m_1}{m_2}{k}$ and we have
\begin{align*}
\mathbb{G}  ={}& (n-m_1) (n-m_2) \FC{n-1}{m_1}{m_2}{k}  &&\text{(non-faulty pair)}\\
      &+ m_1 (n-m_2) \FC{n-1}{m_1-1}{m_2}{k-1} &&\text{($1$-faulty pair)}\\
      &+ (n-m_1) m_2 \FC{n-1}{m_1}{m_2-1}{k-1} &&\text{($2$-faulty pair)}\\
      &+ m_1 m_2 \FC{n-1}{m_1-1}{m_2-1}{k-1}.  &&\text{(both-faulty pair)}
\end{align*}
Notice that
          if $n = k$, then the non-faulty pair case does not apply, as $\FC{n-1}{m_1}{m_2}{k} = 0$, and evaluates to zero.
Likewise, if $b_1 = 0$, then the $1$-faulty pair case does not apply, as $\FC{n-1}{m_1-1}{m_2}{k-1} = 0$, and evaluates to zero; 
          if $b_2 = 0$, then the $2$-faulty pair case does not apply, as $\FC{n-1}{m_1}{m_2-1}{k-1} = 0$, and evaluates to zero;
and, finally, if $b_{1,2} = 0$, then the both-faulty pair case does not apply, as $\FC{n-1}{m_1-1}{m_2-1}{k-1} = 0$, and evaluates to zero.

First, we consider the case in which $n > k$, $b_1 > 0$, $b_2 > 0$, and $b_{1,2} > 0$. Hence, each of the four cases apply and evaluate to non-zero values. We directly apply the induction hypothesis on $\FC{n-1}{m_1}{m_2}{k}$, $\FC{n-1}{m_1-1}{m_2}{k-1}$, $\FC{n-1}{m_1}{m_2-1}{k-1}$, and $\FC{n-1}{m_1-1}{m_2-1}{k-1}$, and obtain
\begin{align*}
\mathbb{G}
    ={}&  (n-m_1) (n-m_2) \frac{m_1!m_2!(n-1-m_1)!(n-1-m_2)!(n-1)!}{b_1!b_2!b_{1,2}!(n-1-k)!}\\
       & + m_1 (n-m_2)     \frac{(m_1-1)!m_2!(n-m_1)!(n-1-m_2)!(n-1)!}{(b_1-1)!b_2!b_{1,2}!(n-1-(k-1))!}\\
       &+ (n-m_1) m_2     \frac{m_1!(m_2-1)!(n-1-m_1)!(n-m_2)!(n-1)!}{b_1!(b_2-1)!b_{1,2}!(n-1-(k-1))!}\\
       &+ m_1 m_2         \frac{(m_1 - 1)!(m_2-1)!(n-m_1)!(n-m_2)!(n-1)!}{b_1!b_2!(b_{1,2}-1)!(n-1-(k-1))!}.
\intertext{We apply $x! = x(x-1)!$ and further simplify and obtain}
\mathbb{G}
    ={}&  \frac{m_1!m_2!(n-m_1)!(n-m_2)!(n-1)!}{b_1!b_2!b_{1,2}!(n-1-k)!}
        + \frac{m_1!m_2!(n-m_1)!(n-m_2)!(n-1)!}{(b_1-1)!b_2!b_{1,2}!(n-k)!}\\
       &+ \frac{m_1!m_2!(n-m_1)!(n-m_2)!(n-1)!}{b_1!(b_2-1)!b_{1,2}!(n-k)!}
        + \frac{m_1!m_2!(n-m_1)!(n-m_2)!(n-1)!}{b_1!b_2!(b_{1,2}-1)!(n-k)!}\\\displaybreak[0]
    ={}&  (n-k)   \frac{m_1!m_2!(n-m_1)!(n-m_2)!(n-1)!}{b_1!b_2!b_{1,2}!(n-k)!}
        + b_1\frac{m_1!m_2!(n-m_1)!(n-m_2)!(n-1)!}{b_1!b_2!b_{1,2}!(n-k)!}\\
       &+ b_2     \frac{m_1!m_2!(n-m_1)!(n-m_2)!(n-1)!}{b_1!b_2!b_{1,2}!(n-k)!}
       + b_{1,2} \frac{m-1!m_2!(n-m_1)!(n-m_2)!(n-1)!}{b_1!b_2!b_{1,2}!(n-k)!}.
\intertext{We have $k = b_1 + b_2 + b_{1,2}$ and, hence, $n = (n-k) + b_1 + b_2 + b_{1,2}$ and we conclude}
\mathbb{G}   ={}&((n-k) + b_1 + b_2 + b_{1,2}) \frac{m_1!m_2!(n-m_1)!(n-m_2)!(n-1)!}{b_1!b_2!b_{1,2}!(n-k)!}\\
    ={}&n\frac{m_1!m_2!(n-m_1)!(n-m_2)!(n-1)!}{b_1!b_2!b_{1,2}!(n-k)!} = \frac{m_1!m_2!(n-m_1)!(n-m_2)!n!}{b_1!b_2!b_{1,2}!(n-k)!}.
\end{align*}

Next, in all other cases, we can repeat the above derivation while removing the terms corresponding to the cases that evaluate to $0$. By doing so, we end up with the expression
\begin{align*}
\mathbb{G} &= \frac{(\left(\sum_{t \in T}\ t\right) m_1! m_2! (n-m_1)! (n-m_2)! (n-1)!}{b_1! b_2! b_{1,2}! (n-k)!}.
\end{align*}
in which $T$ contains the term $(n-k)$ if $n > k$ (the non-faulty pair case applies), the term $b_1$ if $b_1 > 0$ (the $1$-faulty case applies), the term $b_2$ if $b_2 > 0$ (the $2$-faulty case applies), and the term $b_{1,2}$ if $b_{1,2} > 0$ (the both-faulty case applies). As each term $(n-k)$, $b_1$, $b_2$, and $b_{1,2}$ is in $T$ whenever the term is non-zero, we have $\sum_{t \in T}\ t = (n-k) + b_1 + b_2 + b_{1,2} = n$. Hence, we can repeat the steps of the above derivation in all cases, and complete the proof.
\end{proof}

\section{Simplification of the Closed Form of $\FC{n}{m_1}{m_2}{k}$}\label{app:closed_form}
Let $g$ be the expression
\[
b_1!^2 \binom{m_1}{b_1}\binom{n-m_2}{b_1} b_2!^2 \binom{n-m_1}{b_2}\binom{m_2}{b_2}
                   \LM{b_1}{b_2}
                      b_{1,2}!^2 \LM{b_1 + b_2}{b_{1,2}}
                      (n-k)!^2   \LM{k}{n-k},
\]
as stated right above Lemma~\ref{lem:bad_combi}. We will show that $g$ is equivalent to the closed form of $\FC{n}{m_1}{m_2}{k}$, as stated in Lemma~\ref{lem:bad_combi}.
\begin{proof}
We use the shorthands $\mathbf{T}_1 = \binom{m_1}{b_1}\binom{n-m_2}{b_1}$ and $\mathbf{T}_2 = \binom{n-m_1}{b_2}\binom{m_2}{b_2}$, and we have
\begin{align*}
g
                &= b_1!^2 \mathbf{T}_1
                      b_2!^2 \mathbf{T}_2
                      \LM{b_1}{b_2}
                      b_{1,2}!^2
                      \LM{b_1 + b_2}{b_{1,2}}
                      (n-k)!^2
                      \LM{k}{n-k}.
\intertext{We apply Lemma~\ref{lem:list_merge} on terms $\LM{b_1}{b_2}$, $\LM{b_1 + b_2}{b_{1,2}}$, and $\LM{k}{n-k}$, apply $k = b_1 +b_2 + b_{1,2}$, and simplify to derive}
g
                &= b_1!^2 \mathbf{T}_1
                      b_2!^2 \mathbf{T}_2
                      \frac{(b_1 + b_2)!}{b_1!b_2!}
                      b_{1,2}!^2
                      \frac{(b_1 + b_2 + b_{1,2})!}{(b_1 + b_2)!b_{1,2}!}
                      (n-k)!^2
                      \frac{(k + n - k)!}{k!(n-k)!}\\
                 &= b_1! \mathbf{T}_1
                   b_2! \mathbf{T}_2
                   b_{1,2}!
                   (n-k)!
                   n!.
\intertext{Finally, we expand the binomial terms $\mathbf{T}_1$ and $\mathbf{T}_2$, apply $b_{1,2} = m_1 - b_1 = m_2 - b_2$ and $k = m_1 + b_2 = m_2 + b_1$, and simplify to derive}
g
                &= b_1! \frac{m_1!}{b_1!(m_1 - b_1)!}\frac{(n-m_2)!}{b_1!(n-m_2 - b_1)!}
                   b_2! \frac{(n-m_1)!}{b_2!(n-m_1 - b_2)!} \frac{m_2!}{b_2!(m_2 - b_2)!}
                   b_{1,2}!
                   (n-k)!
                   n!\\
                &= \frac{m_1!}{b_{1,2}!}\frac{(n-m_2)!}{b_1!(n-k)!}
                   \frac{(n-m_1)!}{b_2!(n-k)!} \frac{m_2!}{b_{1,2}!}
                   b_{1,2}!
                   (n-k)!
                   n!
                = \frac{m_1!m_2!(n-m_1)!(n-m_2)!n!}{b_1!b_2!b_{1,2}!(n-k)!},
\end{align*}
which completes the proof.
\end{proof}

\section{The Closed Form of $\PT{2f+1}{f}{f}$}\label{app:app_simply}
Here, we shall prove that \[\PT{2f+1}{f}{f} = 4 - \frac{2}{(f+1)} - \frac{f!^2}{(2f)!}.\]
\begin{proof}
By Proposition~\ref{prop:pt_upper} and some simplifications, we have 
\begin{align*}
 \PT{2f+1}{f}{f} &= \frac{1}{(2f+1)!^2} \left(\sum_{k = f}^{2f}\ \frac{2f+1}{2f+1-k} \frac{f!^2(f+1)!^2(2f+1)!}{(k-f)!^2(2f - k)!(2f+1-k)!}\right).\\
\intertext{First, we apply $x! = x(x-1)!$, simplify, and obtain}
\PT{2f+1}{f}{f} &=\frac{f!^2(2f+1)}{(2f+1)!} \left(\sum_{k = f}^{2f}\ \frac{(f+1)!^2}{(k-f)!^2(2f+1-k)!^2}\right)\\
   &=\frac{f!^2}{(2f)!} \left(\sum_{k = 0}^{f}\ \frac{(f+1)!^2}{k!^2(f+1-k)!^2}\right)
    =\frac{f!^2}{(2f)!} \left(\sum_{k = 0}^{f}\ \binom{f+1}{k}^2\right).
\intertext{Next, we apply $\binom{m}{n} = \binom{m}{m-n}$, extend the sum by one term, and obtain}   
\PT{2f+1}{f}{f} &=\frac{f!^2}{(2f)!} \left(\left(\sum_{k = 0}^{f+1}\ \binom{f+1}{k}\binom{f+1}{f+1 - k}\right) - \binom{f+1}{f+1}\binom{f+1}{0}\right).
\intertext{Then, we apply Vandermonde's Identity to eliminate the sum and obtain}
\PT{2f+1}{f}{f} &=\frac{f!^2}{(2f)!} \left(\binom{2f+2}{f+1} - 1\right).
\intertext{Finally, we apply straightforward simplifications and obtain}
\PT{2f+1}{f}{f} &=\frac{f!^2}{(2f)!} \frac{(2f+2)!}{(f+1)!(f+1)!} - \frac{f!^2}{(2f)!}
    =\frac{f!^2}{(2f)!} \frac{(2f)!(2f+1)(2f+2)}{f!^2(f+1)^2} - \frac{f!^2}{(2f)!}\\
   &=\frac{(2f+1)(2f+2)}{(f+1)^2} - \frac{f!^2}{(2f)!}
    =\frac{(2f+2)^2}{(f+1)^2} - \frac{2f+2}{(f+1)^2} - \frac{f!^2}{(2f)!}\\
&=\frac{4(f+1)^2}{(f+1)^2} - \frac{2(f+1)}{(f+1)^2} - \frac{f!^2}{(2f)!}
    =4 - \frac{2}{f+1} - \frac{f!^2}{(2f)!},
\end{align*}
which completes the proof.
\end{proof}

\end{document}